\documentclass[aps,pra,twocolumn,groupedaddress,nofootinbib,notitlepage,showpacs,floatfix]{revtex4-1}

\let\oldsqrt\sqrt
\def\sqrt{\mathpalette\DHLhksqrt}
\def\DHLhksqrt#1#2{%
\setbox0=\hbox{$#1\oldsqrt{#2\,}$}\dimen0=\ht0
\advance\dimen0-0.2\ht0
\setbox2=\hbox{\vrule height\ht0 depth -\dimen0}%
{\box0\lower0.4pt\box2}}

\usepackage{bm,bbm,times,amssymb,amsmath,amsthm,natbib,mathtools,multirow,epsf,latexsym,setspace,array}

\usepackage[normalem]{ulem}
\usepackage{graphicx,graphics}
\usepackage{mathrsfs}
\usepackage{wrapfig}
\usepackage{boxedminipage}
\usepackage{epsfig}
\usepackage{setspace}
\usepackage{subfigure}
\usepackage{hyperref}
\usepackage{dsfont}
\usepackage[all]{xy}
\usepackage[pdftex]{color}
\usepackage{dsfont}

\newtheorem{mytheorem}{Theorem}
\newtheorem{mylemma}{Lemma}

\newcommand{\ignore}[1]{}
\newcommand{\jmdstack}[2]{\genfrac{}{}{0pt}{}{#1}{#2}}
\newcommand{\rmd}{\mathrm{d}}
\newcommand{\identity}{\mathds{1}}
\DeclareMathOperator{\Tr}{Tr}

\DeclareMathOperator{\Ad}{Ad}
\DeclareMathOperator{\sech}{sech}

\DeclareMathOperator*{\esssup}{ess\;sup}
\newcommand{\linOps}{\mathcal{B}(\mathcal{H})}
\newcommand{\HermOps}{\mathcal{M}(\mathcal{H})}

\usepackage{ulem}

\newcommand{\bea} {\begin{eqnarray}}
\newcommand{\eea} {\end{eqnarray}}
\newcommand{\bes} {\begin{subequations}}
\newcommand{\ees} {\end{subequations}}
\newcommand{\bal} {\begin{align}}
\newcommand{\eal} {\end{align}}
\newcommand{\beq}{\begin{equation}}
\newcommand{\eneq}{\end{equation}}
\newcommand{\beqnn}{\begin{equation*}}
\newcommand{\eneqnn}{\end{equation*}}
\newcommand{\beqy}{\begin{eqnarray}}
\newcommand{\eneqy}{\end{eqnarray}}
\newcommand{\beqynn}{\begin{eqnarray*}}
\newcommand{\eneqynn}{\end{eqnarray*}}

\newcommand{\mc}[1]{\mathcal{#1}}

\begin{document}

\title{Analysis of the quantum Zeno effect for quantum control and computation}
\author{Jason M. Dominy$^{(1,4)}$, Gerardo~A. Paz-Silva$^{(1,4)}$, A. T. Rezakhani$^{(1,4,5)}$, and D.
A. Lidar$^{(1,2,3,4)}$}
\affiliation{Departments of $^{(1)}$Chemistry, $^{(2)} $Physics, and $^{(3)}$Electrical
Engineering, and $^{(4)}$Center for Quantum Information Science \&
Technology, University of Southern California, Los Angeles, California
90089, USA\\
$^{(5)}$Department of Physics, Sharif University of Technology, Tehran, Iran }

\begin{abstract}
Within quantum information, many methods have been proposed to avoid or correct the deleterious effects of the environment on a system of interest.  In this work, {expanding on our earlier paper \cite{PRDL:12}}, we evaluate the applicability of the quantum Zeno effect as one such method.  Using the algebraic structure of stabilizer quantum error correction codes as a unifying framework, two open-loop protocols are described which involve frequent non-projective (i.e., weak) measurement of either the full stabilizer group or a minimal generating set thereof.  The effectiveness of the protocols is measured by the distance between the final state under the protocol and the final state of an idealized evolution in which system and environment do not interact.  Rigorous bounds on this metric are derived which demonstrate that, under certain assumptions, a Zeno effect may be realized with arbitrarily weak measurements, and that this effect can protect an arbitrary, unknown encoded state against the environment arbitrarily well.
\end{abstract}

\maketitle
\affiliation{$^{(1)}$Departments of Chemistry, $^{(2)} $Physics, and $^{(3)}$Electrical
Engineering, and $^{(4)}$Center for Quantum Information Science \&
Technology, University of Southern California, Los Angeles, California
90089, USA\\
$^{(5)}$Department of Physics, Sharif University of Technology, Tehran, Iran }


\section{Introduction}
\label{sec:introduction}

Decoherence of a quantum system of interest through interaction with an uncontrolled environment is a key obstacle to realizing practical quantum information processors.  A number of methods have been proposed to help deal with this problem, including error avoidance methods such as decoherence free subspaces \cite{Palma:96, Zanardi:97c,Lidar:PRL98}, closed-loop {suppression} methods such as quantum error correction (QEC) \cite{Shor:1995:R2493, Knill:97b,Gottesman:96}, and open-loop suppression methods such as dynamical decoupling (DD) \cite{Viola:98,Zanardi:1999:77,Viola:99} and the quantum Zeno effect (QZE) \cite{Misra:77,Itano:90,Facchi:05, Facchi:08}.  The typical setting for the QZE is a sequence of frequent projective measurements of an observable.  When the frequency of measurements is high enough, this has the effect of forcing the evolution to remain within the eigenspaces of the observable.  With appropriate choices of the observable, the QZE can be exploited to decouple the system from the environment \cite{Facchi:PRL02, Zanardi:1999:77}.  

In spite of the fact that one method uses fast unitary operations while the other uses frequent measurements, there is a conceptual similarity between DD and QZE protection of quantum states, in that both are feedback-free methods. Indeed, in the ``bang-bang" limit of arbitrarily strong and fast pulses or measurements, it has been shown that DD and the QZE are formally
equivalent \cite{Facchi:03,Facchi:05}.
However, strong projective measurements are an idealization of more realistic, generalized measurements, and likewise, real dynamical decoupling pulses are subject to constraints of finite bandwidth. It has been shown that DD can work to suppress decoherence while allowing for universal quantum computation even with pulses constrained by finite width and repetition rate \cite{Khodjasteh:2008:062310,Khodjasteh:2010:090501,PhysRevLett.105.230503}. Until recently, an analogous result was lacking for the QZE. 

In our earlier paper \cite{PRDL:12} we showed how the assumption of projective measurements can be relaxed and replaced with weak, non-projective, measurements.  Weak measurements extract less information from the system than the corresponding projective measurements, and consequently do not fully collapse the state \cite{PhysRevA.41.11,brun:719}.  
In the present work, significantly expanding on our earlier paper \cite{PRDL:12}, we analyze protocols for realizing a QZE using frequent, weak measurements that are also non-selective, meaning that outcomes are not recorded, with the effect that the state after measurement is the ensemble average over all of the possible outcomes. We will show that these protocols can be used to protect arbitrary, unknown states encoded within a stabilizer QEC code, arbitrarily well. This will be referred to as the weak measurement quantum Zeno effect (WMQZE). Since our protocols involve measuring the stabilizers of QEC codes---a capability that is taken for granted in QEC theory \cite{Nielsen:book,Gaitan:2008:CRC}---but we do not assume that we can observe or use the measured syndrome, our assumptions are weaker than those of QEC, and hence the ability to perform QEC implies the ability to perform our protocols. 

Weak measurements are in some sense the analog of finite bandwidth DD and are both more realistic and more general than strong, projective measurements. They capture a large variety of experimental imperfections and uncertainties \cite{Oreshkov:2005:110409}. However, here, as in our earlier analysis, the measurements are treated as instantaneous, and a generalization to measurements of finite duration is still lacking.

The paper is organized as follows.  Section~\ref{sec:back} gives essential background on weak measurements, reviews the WMQZE protocols introduced in Ref.~\cite{PRDL:12}, and states the main result we prove in the paper: Theorem~\ref{Th:1}, a distance bound quantifying the performance of the protocols. Section~\ref{sec:stabStructure} describes some algebraic structures associated with stabilizer quantum error correction codes and the behavior of weak measurements of the stabilizer elements with respect to these structures.  Section~\ref{sec:boundAnalysis} is concerned with proving Theorem~\ref{Th:1}.  In Section~\ref{sec:parameterTradeOffs}, some trade-offs are considered, between the number of measurements $M$ and the final time $\tau$, as well as between $M$ and the measurement strength $\epsilon$.  Conclusions are presented in Section~\ref{sec:conclusion}.  Several appendices offer additional supporting mathematical details.


\section{Background and statement of main result}
\label{sec:back}

\subsection{Weak Measurements}
Consider a system with Hilbert space $\mathcal{H}_{S}$ coupled to a bath with Hilbert space $\mathcal{H}_B$.  The Hilbert space of the system-bath composite is denoted $\mathcal{H} = \mathcal{H}_{S}\otimes\mathcal{H}_{B}$.  Let $\linOps$ denote the space of bounded linear operators on $\mathcal{H}$.  A measurement of the system can generally be expressed as a \textit{positive operator-valued measure} (POVM), comprising a set of ``measurement operators" $\{M_{j}{\in{\mathcal{B}(\mathcal{H}_{S})\otimes\identity\subset\linOps}}\}$ acting on the system and satisfying the sum rule $\sum_{j}M_{j}^{\dag}M_{j} = \identity$, which map a state $\varrho$ to $\varrho_{j} = M_{j}\varrho M_{j}^{\dag}/p_{j}$ with probability $p_{j} = \Tr(M_{j}\varrho M_{j}^{\dag})$.  The observables that are to be measured in the encoded WMQZE protocols will be elements of the stabilizer group of a given quantum error correcting code (QECC).  As such, they are unitary involutions, i.e., unitary operators $S\in\mathrm{U}(\mathcal{H}_{S})\otimes\identity\subset\mathrm{U}(\mathcal{H})$ such that $S^{2} = \identity$, and therefore have only two possible outcomes (eigenvalues): $\pm 1$.  The weak measurement of such an observable on a state $\varrho$ may be parametrized by the measurement strength $\epsilon$ as \cite{Oreshkov:2005:110409} 
\begin{equation}
\mathcal{P}_{S,\epsilon}(\varrho) = P_{S}(\epsilon)\varrho P_{S}(\epsilon) + P_{S}(-\epsilon)\varrho P_{S}(-\epsilon)\label{eqn:genWeakPOVM}
\end{equation}
where 
\begin{subequations}
\begin{align}
\label{eq:P_S(eps)}
	P_{S}(\epsilon) & :=\alpha_{+}(\epsilon)P_{S} + \alpha_{-}(\epsilon)P_{-S},\\
	\alpha_{\pm}(\epsilon) &= \sqrt{(1\pm\tanh(\epsilon))/2},	
\end{align}
\end{subequations}
and 
\bea
P_{\pm S} := \frac{1}{2}(\identity \pm S)
\eea 
are orthogonal projections onto the $\pm 1$ eigenspaces of $S$.  Since $P_{S}(\epsilon)^{2} + P_{S}(-\epsilon)^{2} = \identity$, the sum rule is satisfied, and $\mathcal{P}_{S,\epsilon}$ is a well-defined POVM for every $\epsilon\in(0,\infty)$.  Moreover, $\mathcal{P}_{S,\epsilon}$ describes a parametrized curve through POVM space that interpolates between projective measurement of $S$ at $\epsilon = \infty$, and no measurement at $\epsilon = 0$.  These weak, non-selective measurements will form the building blocks of the WMQZE protocols to be described herein.  It may also be observed that this 2-term POVM is unitarily equivalent \cite{Nielsen:book} to the 3-term POVM with measurement operators 
\begin{equation}
	M_{1} = \frac{\sqrt{1-\zeta}}{2}(\identity + S),\ M_{2} = \frac{\sqrt{1-\zeta}}{2}(\identity - S),\ M_{3} = \sqrt{\zeta}\identity
\end{equation}
where $\zeta := 2\alpha_{+}(\epsilon)\alpha_{-}(\epsilon) = \sech(\epsilon)$.
This three term POVM may be interpreted as a measurement with a particularly simple classical error, in which, with probability $\zeta$, no measurement takes place, and with probability $1-\zeta$ a projective measurement of $S$ is performed.  It is a completely equivalent description of the non-selective weak measurement $\mathcal{P}_{S,\epsilon}$, although the \textit{selective} measurements corresponding to these two POVMs are not equivalent.

\subsection{The WMQZE Protocols}

Previous WMQZE work applied mostly to particular states \cite{PhysRevLett.97.260402, Xiao2006424,PhysRevB.73.085317,gong:9984}, with some exceptions \cite{Koshino:05}.  In order to protect an arbitrary, unknown $k$-qubit state, as well as to facilitate the analysis that is to come and to allow this WMQZE method to dovetail easily with other protection schemes like QEC, we encode the state into an $[[n,k,d]]$ stabilizer quantum error correcting code (QECC) \cite{Gottesman:96,Nielsen:book}, with stabilizer group $\mathbf{S}=\{S_{i}\}_{i=0}^{Q}$, and where $S_{0}\equiv \identity$.  We assume that the code distance $d\geq 2$, i.e., the code is at least error-detecting, with minimal generating set $\mathbf{\bar{S}=\{}\bar{S}_{i}\}_{i=1}^{\bar{Q}}\subset \mathbf{S}$, where $\bar{Q}=n-k$.  Then every stabilizer element can be uniquely decomposed as $S_{i}=\prod_{\nu =1}^{\bar{Q}}\bar{S} _{\nu }^{r_{i\nu }}$, where ${r_{i\nu }}\in \{0,1\}$, i.e., the stabilizer elements are given by all possible products of the generators, whence $Q+1=2^{\bar{Q}}$.  The encoded initial state $\varrho_{0}$ commutes with all stabilizer elements, and so is supported on the simultaneous $+1$ eigenspace of all the elements of $\mathbf{S}$.  For a given measurement strength, a weak measurement operator $\mathcal{P}_{S,\epsilon}$ may be generated for each $S\in\mathbf{S}$ as in Eq.~\eqref{eqn:genWeakPOVM}.  Since the stabilizer group $\mathbf{S}$ is abelian, the measurements can be performed simultaneously, and simultaneous measurement of the full stabilizer group can be described by the POVM 
\begin{equation}
\label{eq:P-allS}
\mathcal{P}_{\epsilon} := \prod_{S\in\mathbf{S}}\mathcal{P}_{S,\epsilon}.
\end{equation} 
Similarly, we could measure just the generators $\mathbf{\bar{S}}$, whence we define 
\begin{equation}\bar{\mathcal{P}}_{\epsilon} := \prod_{\bar{S}\in\mathbf{\bar{S}}}\mathcal{P}_{\bar{S},\epsilon}.
\label{eq:P-genS}
\end{equation}  
Thus we can now define a \textit{weak stabilizer group measurement protocol} in which $M$ evenly-spaced measurements of the full group are performed over a time interval $[0,\tau]$.  The state of the system-environment composite then evolves as $\left({\mathcal{P}}_{\epsilon}{\mathcal{U}}(\tau /M)\right) ^{M}(\varrho_{SB})$,
\begin{align}
&\left({\mathcal{P}}_{\epsilon}{\mathcal{U}}(\tau /M)\right) ^{M} := \notag \\
& \quad \mathcal{P}_{\epsilon}\mathcal{U}\left(\tau,\tau_{M-1}\right)\mathcal{P}_{\epsilon}\mathcal{U}\left(\tau_{j-1},\tau_{j-2}\right)\cdots \mathcal{P}_{\epsilon}\mathcal{U}(\tau_1,0),
\label{eq:7}
\end{align}
where $\varrho_{SB}$ is the initial state of the system-bath composite, $\{\tau_j = j\tau/M\}_{j=1}^{M}$ are the instants at which the measurements are applied, $\mathcal{U}(t,t')$ is the unitary evolution superoperator of free evolution over $[t,t']$, {i.e., 
\bea
\mathcal{U}(t,t')(\cdot) =\mathcal{T}\exp\Big(\int_{t'}^t \mathcal{L}(s)ds\Big)(\cdot)= U(t,t')(\cdot) U^\dagger(t,t'),\notag \\
\label{eq:supU}
\eea
where $U(t,t')$ is the solution of the differential equation $\frac{d}{dt}{U}(t,t') = -i H(t){U}(t,t')$, with the boundary condition $U(t,t)=\identity$, and $\mathcal{T}$ denotes time-ordering. Here $H(t)$ is the Hamiltonian of the system-bath composite, i.e., $H(t)\in\mathcal{B}(\mathcal{H})$, and the superoperator generator is $\mathcal{L}(t) = -i[H(t),\cdot]$. Note that $\mathcal{P}_{\epsilon}$ in Eq.~\eqref{eq:7}, and more generally $\mathbf{S}$, has non-trivial action on the system only.}

In QEC one measures not the full stabilizer group, but rather its generators, in order to extract an error syndrome \cite{Gottesman:96}. It has been recognized that these syndrome measurements implement a QZE \cite{PhysRevA.54.R1745,PhysRevA.72.012306}.  A \textit{weak stabilizer generator measurement protocol} comprises $M$ evenly-spaced measurements of the generating set over a time interval $[0,\tau]$, so that the state evolves as $\left({\bar{\mathcal{P}}}_{\epsilon}{\mathcal{U}}(\tau /M)\right) ^{M}(\varrho_{0})$,
\begin{align}
&\left(\bar{{\mathcal{P}}}_{\epsilon}{\mathcal{U}}(\tau /M)\right) ^{M} := \notag \\
& \quad \bar{\mathcal{P}}_{\epsilon}\mathcal{U}\left(\tau,\tau_{M-1}\right)\bar{\mathcal{P}}_{\epsilon}\mathcal{U}\left(\tau_{j-1},\tau_{j-2}\right)\cdots \bar{\mathcal{P}}_{\epsilon}\mathcal{U}(\tau_1,0).
\end{align}
This can obviously be an important saving over a full stabilizer group measurement (${\mathcal{P}}_{\epsilon} $).  If the measurement is performed, e.g., by attaching an ancilla for each measured Pauli observable (as in a typical fault-tolerant QEC implementation \cite{Nielsen:book}), then this translates into an exponential saving in the number of such ancillas. We shall consider both protocols in our general development below.

\subsection{Distance Bounds}
To quantify the behavior of these protocols, we use a distance metric as the figure of merit:
\begin{equation}
D[\varrho_{S}(\tau),\varrho_{S}^{0}(\tau)] = \frac{1}{2}\lVert\varrho_{S}(\tau)  - \varrho_{S}^{0}(\tau)\rVert_1,
\label{eq:D}
\end{equation} 
where the norm is the trace norm (sum of the singular values), $\varrho_{S}(\tau) = \Tr_B[\varrho_{SB}(\tau)]$ is the reduced density matrix of the system at time $\tau$ after application of one of the two WMQZE protocols, and $\varrho_{S}^{0}(\tau)= \Tr_B[\varrho^0_{SB}(\tau)]$ is the final state of the system under the idealized circumstance that the system retains its internal Hamiltonian evolution, but does not interact with the environment.  We shall also require the {Schatten $\infty$} operator norm $\lVert\cdot\rVert$ (the maximal singular value). 

The principal results in this paper will be the proof and analysis of the following upper bound on $D[\varrho_{S}(\tau),\varrho_{S}^{0}(\tau)]$.
\begin{mytheorem}
\label{Th:1} 
Assume an arbitrary pure state $\varrho_S = |\psi _{S}\rangle \langle \psi_S|$ 
is encoded into an $[[n,k,d]]$ stabilizer QECC with stabilizer group $\mathbf{S}$. Assume further that the (possibly time-dependent) Hamiltonian $H = H_{S} + H_{B}+ H_{SB}$, where $H_{S}$ and $H_{B}$ represent the system and bath Hamiltonians and $H_{SB}$ the system-bath interaction, is such that $H_{\identity} = H_{S}+H_{B}$ commutes with the stabilizer (i.e., is a linear combination of stabilizer elements and logical operators) and that the interaction term $H_{SB}= \sum_{g\neq \identity}H_{g}$ is a linear combination of error terms detected by the code.  Let $J_{0} = 2\lVert H_{\identity}\rVert$ and $J_{1} = 2\lVert H_{SB}\rVert$, where $J_{0}$ and $J_{1}$ are assumed to be finite.  {Finally,} let $Q=2^{n-k}-1$ and $q=(Q+1)/2$.  Then the stabilizer group measurement protocol 
$\left({\mathcal{P}}_{\epsilon}{\mathcal{U}}(\tau /M)\right) ^{M}$ protects $\varrho_S$ 
up to a deviation that converges to $0$ in the large-$M$ limit:
\begin{align}
D&[\varrho_{S}(\tau), \varrho_{S}^{0}(\tau)]  \leq [1+\Gamma_{\identity}(M)\big]^{M}- \nonumber\\
	& \frac{\Gamma_{-}(M)}{\Gamma_{\identity}(M)}[1+\Gamma_{+}(M)\big]^{M} +  \Gamma_{g}(M)A_{+}(M)\gamma_{+}^{M-1}(M) +\notag \\
&\ \Gamma_{g}(M)A_{-}(M)\gamma_{-}^{M-1}(M)- e^{\tau J_{0}} =:\mathrm{B},
\label{eqn:fullBound}
\end{align}
where the bound $\mathrm{B}$ can be expanded in powers of $1/M$ as
\begin{align}
	& \mathrm{B}= \left[e^{\tau J_{0}}\left(\frac{Q\tau^{2} J_{1}^{2}}{4}\right) + e^{\tau J_{m}}\frac{Q\tau J_{1}}{2}\left(1 + \tau J_{m}\right)\frac{\zeta^{q}}{1-\zeta^{q}}\right]\frac{1}{M} \notag \\
&	+ O\left(\frac{1}{M^{2}}\right),
	\label{eqn:asymptoticbound}
\end{align}
where 
\begin{subequations}
\begin{align}
	\zeta & := \sech(\epsilon) \label{eqn:zeta}\\
	\beta(M) & := \begin{cases}\Gamma_{\identity}(M) & J_{0} \geq J_{1}\\
	\Gamma_{g}(M) & J_{0} \leq J_{1} \end{cases}
	\label{eqn:beta}	\\
	\Gamma_{\identity}(M) &:= \frac{1}{Q+1}e^{\frac{\tau J_{0}}{M}}\left(e^{\frac{\tau Q J_{1}}{M}} + Qe^{-\frac{\tau J_{1}}{M}}\right) - 1\\
	\Gamma_{g}(M)&: = \frac{1}{Q+1}e^{\frac{\tau J_{0}}{M}}\left(e^{\frac{\tau Q J_{1}}{M}} - e^{-\frac{\tau J_{1}}{M}}\right)\\
	\gamma_{\pm}(M) & := \frac{1}{2}\big(1+\beta + (1+Q\beta)\zeta^{q}\big) \pm \notag \\
&\quad \frac{1}{2}\sqrt{\big(1+\beta-(1+Q\beta)\zeta^{q}\big)^{2} + 4 Q \beta^{2}\zeta^{q}}
	\label{eqn:rootspm}\\
	A_{\pm}(M) & := \frac{Q\beta\zeta^{q}(\gamma_{\pm}+\beta) + (1+\beta)\big[(1+\beta)-\gamma_{\mp}\big]}{\beta(\gamma_{\pm}-\gamma_{\mp})}\\
	J_{m} & := \max\{J_{0},J_{1}\}\\
	\Gamma_{+} & := \begin{cases}\Gamma_{\identity}(M) & J_{0}\geq J_{1}\\\Gamma_{g}(M) & J_{0}\leq J_{1}\end{cases}\\  \Gamma_{-} &:=\begin{cases}{\Gamma_{g}(M)} & J_{0}\geq J_{1}\\\Gamma_{\identity}(M) & J_{0}\leq J_{1}\end{cases}
	\label{eqn:coeffspm}.
\end{align}
\end{subequations}
For the generator measurement protocol $\left({\overline{\mathcal{P}}}_{\epsilon }
{\mathcal{U}}(\tau /M)\right)^{M}$, replace 
$q$ by 1 {in Eqs.~(\ref{eqn:asymptoticbound}), (\ref{eqn:rootspm}), and (\ref{eqn:coeffspm}).  In the strong-measurement limit ($\epsilon\to\infty$), both protocols yield the distance bound 
\begin{equation}
D[\varrho _{S}(\tau ),\varrho _{S}^{0}(\tau )] \leq e^{J_0 \tau}\!\left[ \left( \frac{Q e^{-\frac{J_1\tau}{ M}}\!+e^{\frac{J_1 \tau Q }{M}}}{Q+1}\right)^M \!\!\!\! -1 \right].
\label{eq:strong}
\end{equation}}
\end{mytheorem}


\section{Stabilizer QECCs and Induced Structures}

\label{sec:stabStructure}

A large class of QECCs can be described by the stabilizer formalism \cite{Gottesman:96,Nielsen:book}, which we briefly review. A stabilizer $\mathbf{S}$ is an Abelian subgroup of the Pauli group $\mathbf{G} _{n}$ on $n$ qubits that does not contain the element $-\identity$. The Pauli group consists of all possible $n$-fold tensor products of the Pauli matrices $\sigma^{x}\equiv X$, $\sigma ^{y}=Y$, $\sigma ^{z}=Z$ together with the multiplicative factors $\pm 1$, $\pm i$. All elements of $\mathbf{G}_{n}$ are unitary and either Hermitian or skew-Hermitian.  Since $-\identity\notin \mathbf{S}$, all elements of $\mathbf{S}$ must be unitary and Hermitian, and therefore are involutions [$S^{2} = \identity$ for all $S\in\mathbf{S}$].  The stabilizer code $\mathcal{C}$ corresponding to $\mathbf{S}$ is the subspace of all states $|\psi \rangle $ which are invariant under the action of every operator in $\mathbf{S}$ ($S|\psi \rangle =|\psi \rangle $, $\forall S\in \mathbf{S}$). The stabilizer of a code encoding $k$ logical qubits into $n$ physical qubits has $\bar{Q}=n-k$ generators, and $\mathbf{S}$ has $Q=2^{\bar{Q}}$ elements.  A set of errors $\{E_{i}\}$ in $\mathbf{G}_{n}$ is correctable (detectable) by the code if and only if $E_{i}^{\dagger }E_{j}$ ($E_{j}$)\ anticommutes for all $i$ and $j$ (for all $j$) with at least one element of $\mathbf{S}$, or otherwise belongs to $\mathbf{S}$. The normalizer $N(\mathbf{S}):=\{N\ | \ NS = SN\ \forall S\in\mathbf{S}\} \subset \mathbf{G}_{n}$ is the set of logical operations on the code.

We now fix a minimal set $\bar{\mathbf{S}} = \{\bar{S}_{1},\dots, \bar{S}_{\bar{Q}}\}$ of generators of the stabilizer group.  This generating set defines a group isomorphism $B:\mathbb{Z}_{2}^{\bar{Q}}\to\mathbf{S}$ by $B(b_{1},\dots, b_{\bar{Q}}) = \prod_{j=1}^{\bar{Q}}\bar{S}_{j}^{b_{j}}$, where $\mathbb{Z}_{2}=\{0,1\}$ is the additive group of integers mod $2$. The inverse function $B^{-1}(S)$ identifies the subset of generators whose product comprises $S$, namely generator $\bar{S}_j$ participates in the product iff $b_j=1$. Define for each $g\in\mathbf{S}$ a group homomorphism $\sigma_{g}:\mathbf{S}\to\mathbb{Z}_{2}$ by 
\bea
\sigma_{g}(S) := \langle B^{-1}(S), B^{-1}(g)\rangle \pmod 2,
\eea 
where $\langle \cdot, \cdot\rangle$ denotes the dot product of the two binary vectors.  This homomorphism $\sigma_{g}$ then counts (mod $2$) the number of generators shared by $S$ and $g$.  It is symmetric in that 
\bea
\sigma_{S}(g) = \sigma_{g}(S).
\eea

We recall some basic properties of homomorphisms of finite groups \cite{Dummit1999}.  First, a group homomorphism $\phi:G\to H$ is a map satisfying the property $\phi(g_{1}g_{2}) = \phi(g_{1})\phi(g_{2})$.  This implies, in particular, that $\phi(\identity_{G}) = \identity_{H}$.  Both the kernel $K = \phi^{-1}(\identity_{H})\subset G$ and the image $\phi(G)\subset H$ of a homomorphism are subgroups of their respective groups.  If $\phi(g_{1}) = \phi(g_{2})$ then $\phi(g_{1}g_{2}^{-1}) = \identity_{H}$, so if $\phi^{-1}(h)$ is nonempty for some $h\in H$, then $\phi^{-1}(h) = K g$ for any $g\in\phi^{-1}(h)$.  In other words, all non-empty fibers of $\phi$ are cosets of the kernel $K$.  Therefore all non-empty fibers have the same cardinality as the kernel, so either $|\phi^{-1}(h)| = |K|$ or $|\phi^{-1}(h)| = 0$.

\begin{mylemma}
\label{lem:homomorphismProperties}
\bea
\sigma_{S}(g) = 0\ \forall\ g\in\mathbf{S}\ \text{iff}\ S = \identity,
\eea 
and for $S\neq \identity$, 
\bea
|\sigma_{S}^{-1}(0)| = |\sigma_{S}^{-1}(1)| = |\mathbf{S}|/2 = 2^{\bar{Q}-1} = q.
\eea  
Consequently, $\sigma_{S} = \sigma_{S'}$ if and only if $S = S'$.  So $\{(-1)^{\sigma_{S}(\cdot)}\;:\; S\in\mathbf{S}\}$ is the set of all $2^{\bar{Q}}$ complex irreducible representations of $\mathbf{S}$ \cite{Fulton1991}.  
\end{mylemma}
\begin{proof}
Since $B(0,\dots, 0) = \identity$, $B^{-1}(\identity) = (0,\dots, 0)$, which yields zero when dotted with any binary vector $B^{-1}(g)$.  So $\sigma_{\identity}(g) = 0$ for all $g$.  If $S\neq \identity$, then $S = \bar{S}_{j_{1}}\cdots \bar{S}_{j_{k}}$ for some subset of distinct generators $\bar{S}_{j_{1}},\dots, \bar{S}_{j_{k}}$.  Then $\sigma_{S}(\bar{S}_{j_{i}}) = 1$ for all $1\leq i\leq k$.  So $\sigma_{S}(g) = 0$ for all $g\in\mathbf{S}$ if and only if $S = \identity$.  In the case $S\neq \identity$, the image of $\sigma_{S}$ is all of $\mathbb{Z}_{2}$.  Then, the fibers $K = \sigma_{S}^{-1}(0)$ and $\sigma_{S}^{-1}(1)$ are both nonempty and therefore are cosets of the kernel $K$ \cite{Dummit1999} and partition the group $\mathbf{S}$.  Thus, $|K| = |\sigma_{S}^{-1}(0)| = |\sigma_{S}^{-1}(1)|$ and $|\mathbf{S}| = |\sigma_{S}^{-1}(0)\cup\sigma_{S}^{-1}(1)| = 2|K|$.  Finally, if $\sigma_{S}(g) = \sigma_{S'}(g)$ for all $g\in\mathbf{S}$, then \begin{equation}\sigma_{SS'}(g) = \sigma_{g}(SS') = \sigma_{g}(S) + \sigma_{g}(S') = \sigma_{S}(g) + \sigma_{S'}(g) = 0\end{equation}
for all $g\in\mathbf{S}$, which, by the arguments above, holds if and only if $SS' = \identity$, i.e., if and only if $S = S'$.
\end{proof}

\subsection{Isotypical Decompositions}

The fact that the homomorphisms $\{(-1)^{\sigma_{g}(\cdot)}\;:\; g\in\mathbf{S}\}$ are the irreducible representations of $\mathbf{S}\simeq \mathbb{Z}_{2}^{\bar{Q}}$ leads to a natural and well known orthogonal decomposition of the state space $\mathcal{H}$ into code subspaces \cite{Gottesman:96}.

\begin{mylemma}
\label{lem:HilbSpaceIsotypical}
The isomorphism $B:\mathbb{Z}_{2}^{\bar{Q}}\to \mathbf{S}$ is a faithful unitary representation of $\mathbb{Z}_{2}^{\bar{Q}}$ in terms of operators on $\mathcal{H} = \mathcal{H}_{S}\otimes\mathcal{H}_{B}$, where $\mathcal{H}_{S}\simeq \mathbb{C}^{2^{n}}$.  There is then a unique isotypical decomposition \cite{Fulton1991} of $\mathcal{H}$ into subspaces
\begin{equation}
\mathcal{H} = \bigoplus_{g\in\mathbf{S}}V_{g},\quad
V_{g} = \hat{V}_{g}^{\oplus a_{g}},\quad a_g=2^k\dim(\mathcal{H}_{B}),
\end{equation}
 where each $V_{g}$ is an invariant subspace of the representation $B$ and the projection of $B$ onto any one-dimensional subspace $\hat{V}_{g}$ thereof is the irreducible representation $(-1)^{\sigma_{g}(\cdot)}$, i.e., for any $|\psi\rangle \in V_{g}$, $S|\psi\rangle = (-1)^{\sigma_{g}(S)}|\psi\rangle$.  Since each $\hat{V}_{g}$ is one-dimensional (because $\mathbf{S}\simeq \mathbb{Z}_{2}^{\bar{Q}}$ is abelian), the exponent $a_{g}$ is the dimension of the subspace $V_{g}$.  With $\Tr(\identity) = 2^{n}$ and all other elements of $\mathbf{S}$ traceless, $a_{g} = 2^{n-\bar{Q}}\dim(\mathcal{H}_{B}) = 2^{k}\dim(\mathcal{H}_{B})$ for all $g$.  The $2^{k}\dim(\mathcal{H}_{B})$-dimensional subspaces $V_{g}$ are all orthogonal, and $V_{\identity}$ is the subspace stabilized by $\mathbf{S}$.
\end{mylemma}
\begin{proof}
The isotypical decomposition is a standard result in representation theory, following from Schur's lemma \cite{Fulton1991}.  If $|\psi\rangle\in V_{g},$ then $S|\psi\rangle = (-1)^{\sigma_{g}(S)}|\psi\rangle$ because $(-1)^{\sigma_{g}(\cdot)}$ is the irreducible representation associated to $V_{g}$.  By \cite[Corollary 2.16]{Fulton1991}, 
\begin{align}
	a_{g} & = \frac{1}{2^{\bar{Q}}}\sum_{S\in\mathbf{S}}\Tr(S)(-1)^{\sigma_{g}(S)} = 2^{n-\bar{Q}}\dim(\mathcal{H}_{B}) \notag\\
	& = 2^{k}\dim(\mathcal{H}_{B})
\end{align}
using the fact that $S = \identity$ has $\Tr(\identity) = 2^{n}$, and all other $S\in\mathbf{S}$ have zero trace.  If $|\psi_{g}\rangle \in V_{g}$ and $|\psi_{h}\rangle\in V_{h}$, then since $\mathbf{S}\subset\mathrm{U}(\mathcal{H})$, 
\bea
\langle \psi_{g}|\psi_{h}\rangle &=& \langle S\psi_{g}|S\psi_{h}\rangle = (-1)^{\sigma_{g}(S)+\sigma_{h}(S)}\langle \psi_{g}|\psi_{h}\rangle \notag \\
&=& (-1)^{\sigma_{gh}(S)}\langle \psi_{g}|\psi_{h}\rangle
\eea
for all $S\in\mathbf{S}$.  If $gh\neq \identity$, i.e., if $g\neq h$, then by Lemma \ref{lem:homomorphismProperties}, $|\sigma_{gh}^{-1}(1)| = |\mathbf{S}|/2$, so there exists $S\in\mathbf{S}$ such that $\sigma_{gh}(S) = 1$.  Therefore $\langle \psi_{g}|\psi_{h}\rangle$ must be zero and the subspaces $V_{g}$ form an orthogonal decomposition of $\mathcal{H}$.
\end{proof}

In the language of QEC \cite{Gottesman:96}, each of the subspaces $V_g$ can be thought of as encoding $k$ qubits, but only $V_{\identity}$ is stabilized by $\mathbf{S}$ [i.e., $S|\psi\rangle = (-1)^{\sigma_{g}(S)}|\psi\rangle = |\psi\rangle\ \forall S\in\mathbf{S}\ $ iff $\sigma_{g}(S) = 0 \ \forall S\in\mathbf{S}\ \Leftrightarrow\ g=\identity$ by Lemma~\ref{lem:homomorphismProperties}].  Hence $V_{\identity}$ is typically chosen as the stabilizer QECC. With this choice, the remaining isotypical subspaces are interpreted as ``syndrome'' subspaces, where $g$ labels the syndrome. Namely, after an error that is detectable by the code takes place, it maps $V_{\identity}$ to one of the other subspaces $V_g$. A measurement of all the stabilizer generators reveals the label $g$, in that $g = B(b_1,\dots,b_{\bar{Q}})$, where $b_j\in \mathbb{Z}_2$ is $0$ ($1$) if the measurement of generator $\bar{S}_j$ yielded eigenvalue $+1$ ($-1$), with $j\in\{1,\dots,\bar{Q}\}$.

Using the group action of conjugation by $\mathbf{S}$, the space $\linOps$ of all linear operators (complex matrices) on $\mathcal{H} = \mathcal{H}_{S}\otimes\mathcal{H}_{B}$ and the space $\HermOps$ of Hermitian operators may be similarly decomposed into isotypical subspaces indexed by $\mathbf{S}$ which are orthogonal under any inner product invariant under conjugation by $\mathbf{S}$ (such as the Hilbert-Schmidt inner product).  Then 
\begin{equation}
\linOps = \bigoplus_{g\in\mathbf{S}}W_{g}^{\mathbb{C}} \quad \HermOps = \bigoplus_{g\in\mathbf{S}}W_{g},\label{eqn:isotypicalMatrixDecompositions}
\end{equation}
where, for any $g\in\mathbf{S}$, the subspace $W_{g}^{\mathbb{C}}$ (respectively $W_{g}$) is the space of all operators (resp. Hermitian operators) $A_{g}$ with the defining property that they satisfy $SA_{g}S = (-1)^{\sigma_{g}(S)}A_{g}$ for all $S\in\mathbf{S}$.  From this description, it is clear that $W_{g}\subset W_{g}^{\mathbb{C}}$ for all $g\in\mathbf{S}$.  For more on these decompositions, see Appendix \ref{app:isotypicalDecompositions}.
\begin{mylemma}
The isotypical decompositions in Eq.~\eqref{eqn:isotypicalMatrixDecompositions} impart to $\linOps$ the structure of an $\mathbf{S}$-graded associative algebra \cite{Nijenhuis1966} and to $\HermOps$ the structure of an $\mathbf{S}$-graded Lie algebra (under the Lie bracket $A,B\mapsto i[A,B]$).
\label{lem:gradedAlgebras}
\end{mylemma}
\begin{proof}
For any $g,h\in\mathbf{S}$ and any $A_{g}\in W_{g}^{\mathbb{C}}$ and $A_{h}\in W_{h}^{\mathbb{C}}$, 
\bea
SA_{g}A_{h}S &=& S A_{g}SSA_{h}S = (-1)^{\sigma_{g}(S) + \sigma_{h}(S)}A_{g}A_{h} \notag \\
&=& (-1)^{\sigma_{gh}(S)}A_{g}A_{h},
\eea
so that $A_{g}A_{h}\in W_{gh}^{\mathbb{C}}$, and therefore $W_{g}^{\mathbb{C}}W_{h}^{\mathbb{C}}\subset W_{gh}^{\mathbb{C}}$.  Likewise for any $A_{g}\in W_{g}$ and $A_{h}\in W_{h}$,
\bea
S\big(i[A_{g},A_{h}])S &=& i[S A_{g}S, SA_{h}S] \notag \\
&=& (-1)^{\sigma_{g}(S) + \sigma_{h}(S)}i[A_{g},A_{h}] \notag \\
&=& (-1)^{\sigma_{gh}(S)}i[A_{g},A_{h}],
\eea
so that $i[A_{g},A_{h}]\in W_{gh}$, and therefore $i[W_{g},W_{h}]\subset W_{gh}$.
\end{proof}

Finally, we can define the orthogonal projections into these isotypical subspaces as follows.
\begin{mylemma}
For any $g\in\mathbf{S}$, the operator $\hat{\mathcal{P}}_{g}:\linOps\to W_{g}^{\mathbb{C}}$ defined by 
\begin{equation}
\hat{\mathcal{P}}_{g}(A) := \frac{1}{|\mathbf{S}|}\sum_{S\in\mathbf{S}}(-1)^{\sigma_{g}(S)}SAS
\end{equation}
is the orthogonal projection into the subspace $W_{g}^{\mathbb{C}}$.  Restricted to $\HermOps$, this same operator defines the orthogonal projection into $W_{g}$.
\end{mylemma}
\begin{proof}
For any $A\in \linOps$, and any $g, S'\in\mathbf{S}$ 
\begin{align}
	S'\big(\hat{\mathcal{P}}_{g}(A)\big)S' & = \frac{1}{|\mathbf{S}|}\sum_{S\in\mathbf{S}}(-1)^{\sigma_{g}(S)}S'SAS'S \notag \\
	& = \frac{1}{|\mathbf{S}|}\sum_{S\in\mathbf{S}}(-1)^{\sigma_{g}(S'S)}SAS \notag \\
	& = (-1)^{\sigma_{g}(S')}\frac{1}{|\mathbf{S}|}\sum_{S\in\mathbf{S}}(-1)^{\sigma_{g}(S)}SAS \notag \\
	& = (-1)^{\sigma_{g}(S')}\hat{\mathcal{P}}_{g}(A),
\label{eq:S'PS'}
\end{align}
so that $\hat{\mathcal{P}}_{g}(A)\in W_{g}^{\mathbb{C}}$.  Moreover, we find
\begin{equation}
\sum_{g\in\mathbf{S}}\hat{\mathcal{P}}_{g}(A) = \frac{1}{|\mathbf{S}|}\sum_{S\in\mathbf{S}}\left(\sum_{g\in\mathbf{S}}(-1)^{\sigma_{g}(S)}\right)SAS = A,
\label{eq:projection}
\end{equation}
since by Lemma \ref{lem:homomorphismProperties}, $\sum_{g\in\mathbf{S}}(-1)^{\sigma_{g}(S)} = 0$ for $S\neq \identity$ and equals $|\mathbf{S}|$ when $S = \identity$.  Therefore, since the subspaces $W_{g}^{\mathbb{C}}$ are mutually orthogonal, the operators $\hat{\mathcal{P}}_{g}$ are orthogonal projections.  Finally observe that if $A$ is Hermitian, then $\hat{\mathcal{P}}_{g}(A)$ is Hermitian as well, so $\hat{\mathcal{P}}_{g}$ also describes the orthogonal projection $\HermOps\mapsto W_{g}$.
\end{proof}
Note that Eq.~\eqref{eq:S'PS'} shows that $\hat{\mathcal{P}}_{g}(A)$ coincides with the defining property of the operators $A_g$ (Hermitian or not), so that we can equivalently define $A_{g} := \hat{\mathcal{P}}_{g}(A)$.

\subsection{Measurement Operators}

With respect to these isotypical decompositions, the actions of the measurement operators defined in Section \ref{sec:introduction} take particularly simple forms.  Recalling Eqs.~\eqref{eqn:genWeakPOVM} and \eqref{eq:P_S(eps)}, for any $S\in\mathbf{S}$ and $\epsilon>0$, the effect of the weak measurement of the stabilizer $S$ is given by 
\bea
\mathcal{P}_{S,\epsilon}(A) &:=& P_{S}(\epsilon)AP_{S}(\epsilon) + P_{S}(-\epsilon)A P_{S}(-\epsilon) \\
&=& \sum_{b=\pm 1}\sum_{s_{1}, s_{2} = \pm}\alpha_{s_{1}}(b\epsilon)\alpha_{s_{2}}(b\epsilon)P_{s_{1}S}A P_{s_{2}S}\notag
\eea

for any $A\in\linOps$.  Then
\begin{mylemma}
For any $g,S\in\mathbf{S}$, any $A_{g}\in W_{g}^{\mathbb{C}}$, and any $\epsilon >0$, 
\begin{equation}
\mathcal{P}_{S,\epsilon}(A_{g}) = \zeta^{\sigma_{S}(g)}A_{g},
\label{eq:P_Seps}
\end{equation}
where $\zeta = \sech(\epsilon)$, and therefore, using Eq.~\eqref{eq:projection},
\bea
\mathcal{P}_{S,\epsilon} &{=}& \sum_{g\in\mathbf{S}}\zeta^{\sigma_{S}(g)}\hat{\mathcal{P}}_{g} = \sum_{g\in \sigma_{S}^{-1}(0)}\hat{\mathcal{P}}_{g} + \zeta\sum_{g\in\sigma_{S}^{-1}(1)}\hat{\mathcal{P}}_{g} \notag \\
&{=}& \mathcal{P}_{S,\infty} + \zeta(\identity - \mathcal{P}_{S,\infty}) = (1-\zeta)\mathcal{P}_{S,\infty} + \zeta\identity.\notag \\
\label{eq:P_Seps-1}
\eea
\end{mylemma}
\begin{proof}
Observe first that, from the definition of $W_{g}^{\mathbb{C}}$, 
\bea
A_{g}P_{s_{2}S} &=& \frac{1}{2}A_{g}\big(\identity + s_{2}S\big) = \frac{1}{2}\Big(\identity + s_{2}(-1)^{\sigma_{S}(g)}S\Big)A_{g} \notag \\
&=& P_{s_{2}(-1)^{\sigma_{S}(g)}S}A_{g},
\eea
so that, using the facts that $P_{\pm S}^{2} = P_{\pm S}$, $P_{S}P_{-S} = P_{-S}P_{S} = 0$ [in Eq.~\eqref{eq:29c}], $P_{S} + P_{-S} = \identity$, $\alpha_{\pm}^{2}(\epsilon) + \alpha_{\pm}^{2}(-\epsilon) = 1$, and $\alpha_{+}(\epsilon)\alpha_{-}(\epsilon) + \alpha_{+}(-\epsilon)\alpha_{-}(-\epsilon) = \sech(\epsilon) = \zeta$ [in Eq.~\eqref{eq:29e}],
\begin{subequations}
\begin{align}
	\mathcal{P}_{S,\epsilon}(A_{g}) &= \sum_{b=\pm 1}\sum_{s_{1}, s_{2} = \pm}\alpha_{s_{1}}(b\epsilon)\alpha_{s_{2}}(b\epsilon)\times  \notag \\
	&\qquad\qquad P_{s_{1}S} P_{s_{2}(-1)^{\sigma_{S}(g)}S}A_{g}\\
	&= \sum_{b=\pm 1}\sum_{s_{1}, s_{2} = \pm}\alpha_{s_{1}}(b\epsilon)\alpha_{s_{2}(-1)^{\sigma_{S}(g)}}(b\epsilon)\times\notag \\
	&\qquad\qquad  P_{s_{1}S} P_{s_{2}S}A_{g}\\
	\label{eq:29c}
	&= \sum_{b=\pm 1}\sum_{s = \pm}\alpha_{s}(b\epsilon)\alpha_{s(-1)^{\sigma_{S}(g)}}(b\epsilon)P_{sS}A_{g}\\
	&= \begin{cases}\sum_{b=\pm 1}\sum_{s = \pm}\alpha_{s}^{2}(b\epsilon)P_{sS}A_{g} \\
	\textrm{if}\ \ \sigma_{S}(g) = 0\\ \sum_{b=\pm 1}\alpha_{+}(b\epsilon)\alpha_{-}(b\epsilon)\sum_{s = \pm}P_{sS}A_{g} \\
	 \textrm{if}\ \ \sigma_{S}(g) = 1\end{cases}\\
	\label{eq:29e}
	&= \begin{cases}A_{g} & \sigma_{S}(g) = 0\\ \zeta A_{g} & \sigma_{S}(g) = 1\end{cases},
\end{align}
\end{subequations}
which proves Eq.~\eqref{eq:P_Seps}. {Equation~\eqref{eq:P_Seps-1} follows from the observation that $\sum_{g\in \sigma_{S}^{-1}(0)}\hat{\mathcal{P}}_{g}$ is the projection into the subspace of operators that commute with $S$ {(the commutant or centralizer of $S$)}, which is precisely the strong (von Neumann) {non-selective} measurement of $S$, $\mathcal{P}_{S,\infty}$.}
\end{proof}
We are now ready to see the effect of the complete POVM defined in Eqs.~\eqref{eq:P-allS} and \eqref{eq:P-genS}.
\begin{mylemma}
\label{lem:measurementFullGroup}
For any $g\in\mathbf{S}$,  $A_{g}\in W_{g}^{\mathbb{C}}$, and $\epsilon>0$, the weak measurement $\mathcal{P}_{\epsilon}$ of the full stabilizer group has the effect
\begin{equation}
\mathcal{P}_{\epsilon}(A_{g}) = \begin{cases}A_{g} & g = \identity\\\zeta^{q}A_{g} & \text{else}.\end{cases}
\label{eq:P_eps}
\end{equation}
so that the weak measurement may be written as
\bea
\mathcal{P}_{\epsilon} &{=}& \hat{\mathcal{P}}_{\identity} + \zeta^{q}\sum_{g\neq\identity}\hat{\mathcal{P}}_{g} = \hat{\mathcal{P}}_{\identity} + \zeta^{q}(\identity-\hat{\mathcal{P}}_{\identity}) \notag \\
&{=}& \big(1-\zeta^{q}\big)\hat{\mathcal{P}}_{\identity} + \zeta^{q}\identity = \big(1-\zeta^{q}\big)\mathcal{P}_{\infty} + \zeta^{q}\identity.
\label{eq:calPeps}
\eea
\end{mylemma}
\begin{proof}
Since $\hat{\mathcal{P}}_{g}\hat{\mathcal{P}}_{h} = 0$ for $g\neq h$,
\bea
\mathcal{P}_{\epsilon} &=& \prod_{S\in\mathbf{S}}\mathcal{P}_{S,\epsilon}  = \prod_{S\in\mathbf{S}}\sum_{g_{S}\in\mathbf{S}}\zeta^{\sigma_{S}(g_{S})}\hat{\mathcal{P}}_{g_{S}} = \sum_{g\in\mathbf{S}}\prod_{S\in\mathbf{S}}\zeta^{\sigma_{g}(S)}\hat{\mathcal{P}}_{g} \notag \\
&=& \sum_{g\in\mathbf{S}}\zeta^{|\sigma_{g}^{-1}(1)|}\hat{\mathcal{P}}_{g},
\eea
and from Lemma \ref{lem:homomorphismProperties}, \begin{equation}|\sigma_{g}^{-1}(1)| = \begin{cases}0 & g = \identity\\ q & \text{else}.\end{cases} ,
\end{equation}
which proves Eq.~\eqref{eq:P_eps}. Equation~\eqref{eq:calPeps} follows from the same reasoning as used for Eq.~\eqref{eq:P_Seps-1} in the previous lemma.
\end{proof}

\begin{mylemma}
\label{lem:measurementGenerators}
For any $g\in\mathbf{S}$,  $A_{g}\in W_{g}^{\mathbb{C}}$, and $\epsilon>0$, the weak measurement $\overline{\mathcal{P}}_{\epsilon}$ of the generators of the stabilizer group has the effect
 \begin{equation}\overline{\mathcal{P}}_{\epsilon}(A_{g}) = \zeta^{|\bar{\mathbf{S}}\cap \sigma_{g}^{-1}(1)|}A_{g}\end{equation}
so that the generators-only weak measurement may be written as 
\bea
\overline{\mathcal{P}}_{\epsilon} &=& \hat{\mathcal{P}}_{\identity} + \sum_{g\neq\identity}\zeta^{|\bar{\mathbf{S}}\cap \sigma_{g}^{-1}(1)|}\hat{\mathcal{P}}_{g} \notag \\
&=& \hat{\mathcal{P}}_{\identity} + \sum_{c=1}^{\bar{Q}}\zeta^{c}\sum_{g\in\{h\in\mathbf{S}\;:\;|\bar{\mathbf{S}}\cap \sigma_{h}^{-1}(1)| = c\}}\hat{\mathcal{P}}_{g} \notag \\
&=& \hat{\mathcal{P}}_{\identity} + \sum_{c=1}^{\bar{Q}}\zeta^{c}\sum_{1\leq j_{1}<\dots< j_{c}\leq \bar{Q}}\hat{\mathcal{P}}_{\bar{S}_{j_{1}}\cdots \bar{S}_{j_{c}}}.
\eea
\end{mylemma}
\begin{proof}
Since $\hat{\mathcal{P}}_{g}\hat{\mathcal{P}}_{h} = 0$ for $g\neq h$,
\bea
\overline{\mathcal{P}}_{\epsilon} &=& \prod_{\bar{S}\in\bar{\mathbf{S}}}\mathcal{P}_{\bar{S},\epsilon}  = \prod_{\bar{S}\in\bar{\mathbf{S}}}\sum_{g_{\bar{S}}\in\mathbf{S}}\zeta^{\sigma_{\bar{S}}(g_{\bar{S}})}\hat{\mathcal{P}}_{g_{\bar{S}}} {=} \sum_{g\in\mathbf{S}}\prod_{\bar{S}\in\bar{\mathbf{S}}}\zeta^{\sigma_{g}(\bar{S})}\hat{\mathcal{P}}_{g} \notag \\
&{=}& \sum_{g\in\mathbf{S}}\zeta^{|\bar{\mathbf{S}}\cap\sigma_{g}^{-1}(1)|}\hat{\mathcal{P}}_{g},
\eea
and from Lemma \ref{lem:homomorphismProperties} and the definition of $\sigma_{g}$, 
\begin{equation}
\label{eq:gen-power}
|\bar{\mathbf{S}}\cap\sigma_{g}^{-1}(1)|
\begin{cases} 
= 0 & g = \identity\\ \in\{1,\dots,\bar{Q}\} & \text{else}.
\end{cases}
\end{equation}
\end{proof}

We note that the result of Lemma~\ref{lem:measurementGenerators} is stronger than that reported in our earlier work \cite{PRDL:12}, where we used the lower bound $1$ for $g\neq\identity$ in place of Eq.~\eqref{eq:gen-power}.


\section{Analysis of the Distance Upper Bound}
\label{sec:boundAnalysis}

In this section, we analyze the behavior of the distance $D[\varrho_{S}(\tau), \varrho_{S}^{0}(\tau)]$ [Eq.~\eqref{eq:D}] and show that it converges to $0$ in the limit of large numbers of measurements.  The Hamiltonian is orthogonally decomposed as $H(t) = \sum_{g\in\mathbf{S}}H_{g}(t)$, where $H_{g}(t)$ is the component of $H(t)$ lying in the isotypical (error syndrome) subspace $W_{g}$.  This yields the superoperators $\mathcal{L}_{g} = i[H_{g}, \cdot]$.  We also write $H_{SB} = \sum_{g\neq\identity} H_{g}$, and define $J_{0} = 2\|H_{\identity}\|_{\infty}$ and $J_{1} = 2\|H_{SB}\|_{\infty}$, where $\|\cdot\|_{\infty}$ denotes the $L_{\infty}$ norm, i.e., $\|H_{\identity}\|_{\infty} = \esssup_{t\in[0,\tau]} \|H_{\identity}(t)\|$ in which $\esssup$ denotes the essential supremum and $\|\cdot\|$ denotes the Schatten $\infty$ norm, i.e. the maximum singular value.  This guarantees that $\|\mathcal{L}_{\identity}\|_{\infty} \leq J_{0}$ and $\|\mathcal{L}_{g}\|_{\infty}\leq J_{1}$ for all $g\neq \identity$ in $\mathbf{S}$.  These finite bound conditions may also be shown to imply rapid decay of the noise spectrum at high frequencies, guaranteeing an effective spectral cutoff; conversely, an insufficiently rapidly decaying noise spectrum implies that our finite bound conditions are not satisfied (see Appendix \ref{app:correlationFunction}).  

Let $\mathbb{N} = \{1,2,3,\dots\}$ and $\mathbb{N}_{0} = \{0,1,2,\dots\}$ and observe that the unitary superoperator \eqref{eq:supU} describing the joint system-bath evolution between successive measurements can be written as
\begin{subequations}
\begin{align}
\mathcal{U}\left(\frac{j}{M}\tau, \frac{j+1}{M}\tau\right) & = \identity + \int_{\frac{j}{M}\tau}^{\frac{j+1}{M}\tau}\mathcal{L}(t_{1})\,\rmd t_{1}\notag \\
& + \int_{\frac{j}{M}\tau}^{\frac{j+1}{M}\tau}\int_{\frac{j}{M}\tau}^{t_{1}}\mathcal{L}(t_{1})\mathcal{L}(t_{2})\,\rmd t_{2}\,\rmd t_{1} + \dots\\
& = \sum_{k=0}^{\infty}\sum_{\vec{\alpha}\in\mathbf{S}^{k}}\mathfrak{L}_{j}^{k}(\vec{\alpha})
\end{align}
\end{subequations}
by the Dyson expansion and by the isotypical decomposition $\mc{L} = \sum_{g\in\mathbf{S}} \mc{L}_g$, where $\mathbf{S}^{k}$ is the set of all $k$-tuples of stabilizer group elements, $\mathfrak{L}_{j}^{0} = \identity$, and for any $k>0$ and any $\vec{\alpha}\in\mathbf{S}^{k}$, 
\bea
\mathfrak{L}_{j}^{k}(\vec{\alpha}) &:= \int_{\frac{j-1}{M}\tau}^{\frac{j}{M}\tau}\int_{\frac{j-1}{M}\tau}^{t_{1}}\cdots\int_{\frac{j-1}{M}\tau}^{t_{k-1}}\mathcal{L}_{\alpha_{1}}(t_{1})\mathcal{L}_{\alpha_{2}}(t_{2})\cdots\notag\\
&\cdots\mathcal{L}_{\alpha_{k}}(t_{k})\,\rmd t_{k} \cdots\rmd t_{2}\,\rmd t_{1}.
\eea
Lemmas \ref{lem:gradedAlgebras} and \ref{lem:measurementFullGroup} imply that 
\bea
&\mathcal{P}_{\epsilon}^{i_{j}}\mathfrak{L}_{M+1-i_{1}-\dots -i_{j}}^{l_{j}}(\vec{\alpha}^{j})\cdots\mathfrak{L}_{M+1-i_{1}-\dots -i_{\eta}}^{l_{\eta}}(\vec{\alpha}^{\eta})(\varrho_{SB}) =\notag \\
&\  \zeta^{i_{j}\mu_{j}}\mathfrak{L}_{M+1-i_{1}-\dots -i_{j}}^{l_{j}}(\vec{\alpha}^{j})\cdots\mathfrak{L}_{M+1-i_{1}-\dots -i_{\eta}}^{l_{\eta}}(\vec{\alpha}^{\eta})(\varrho_{SB}),\notag \\
\eea
where $\mu_{j} = \upsilon[\big(\alpha_{1}^{j}\cdots\alpha_{l_{j}}^{j}\big)\cdots \big(\alpha_{1}^{\eta}\cdots\alpha_{l_{\eta}}^{\eta}\big)]$ and $\upsilon:\mathbf{S}\to\{0,1\}$ is defined by \begin{equation}
	\upsilon(g) = 0 \text{ for } g\neq \identity \text{ and }\upsilon(\identity) = 1.
\end{equation}
This follows from the fact, implied by Lemma \ref{lem:gradedAlgebras}, that a composition of Hamiltonian superoperators $\mathcal{L}_{\alpha_{1}}\cdots\mathcal{L}_{\alpha_{k}}$ will map the initial density matrix (an element of the isotypical space $W_{\identity}$) to the isotypical space $W_{\alpha_{1}\cdots\alpha_{k}}$, and that this space, by Lemma \ref{lem:measurementFullGroup}, determines the action of the measurement $\mathcal{P}_{\epsilon}$.  Then it is found that
\begin{subequations}
\begin{align}
	\varrho_{SB}(\tau) & = \mathcal{P}_{\epsilon }\mathcal{U}\left(\frac{M-1}{M}\tau, \tau\right)\mathcal{P}_{\epsilon}\mathcal{U}\left(\frac{M-2}{M}\tau, \frac{M-1}{M}\tau\right)\notag \\
	&\quad \cdots\mathcal{P}_{\epsilon}\mathcal{U}\left(0, \frac{1}{M}\tau\right)\varrho _{SB}= \varrho_{SB} +\notag \\
	& \sum_{\eta=1}^{M}\sum_{l_{1},\dots,l_{\eta} = 1}^{\infty}\sum_{\jmdstack{\vec{i}\in\mathbb{N}^{\eta}}{\|\vec{i}\|_{1}\leq M}}\sum_{\jmdstack{\vec{\alpha}^{j}\in\mathbf{S}^{l_{j}}}{j=1,\dots,\eta}}\mathcal{P}_{\epsilon}^{i_{1}}\mathfrak{L}_{M-i_{1}+1}^{l_{1}}(\vec{\alpha}^{1})\cdots\notag \\
	&\quad \cdots \mathcal{P}_{\epsilon}^{i_{\eta}}\mathfrak{L}_{M+1-i_{1}-\dots -i_{\eta}}^{l_{\eta}}(\vec{\alpha}^{\eta})(\varrho_{SB}) \label{eqn:dysonWithProjectors}\\
	& = \varrho_{SB} + \sum_{\eta=1}^{M}\sum_{l_{1},\dots,l_{\eta} = 1}^{\infty}\sum_{\jmdstack{\vec{i}\in\mathbb{N}^{\eta}}{\|\vec{i}\|_{1}\leq M}}\sum_{\jmdstack{\vec{\alpha}^{j}\in\mathbf{S}^{l_{j}}}{j=1,\dots,\eta}}\zeta^{\frac{Q+1}{2}\big(\vec{i}\cdot\vec{\mu}\big)}\times\notag \\
	&\quad \mathfrak{L}_{M-i_{1}+1}^{l_{1}}(\vec{\alpha}^{1})\cdots\mathfrak{L}_{M+1-i_{1}-\dots -i_{\eta}}^{l_{\eta}}(\vec{\alpha}^{\eta})(\varrho_{SB}),
\end{align}
\end{subequations}
where $\vec{\mu}\in\{0,1\}^{\eta}$, $\mu_{j} = \upsilon(g_{\eta} g_{\eta-1}\cdots g_{j})$, and $g_{j} = \alpha_{1}^{j}\cdots \alpha_{l_{j}}^{j}$.  Then
\begin{align}
	D&[\varrho_{S}(\tau), \varrho_{S}^{0}(\tau)] = \frac{1}{2}\left\Vert \text{Tr}_{B}\left[ \varrho_{SB}(\tau) - \varrho_{SB}^{0}(\tau)\right]\right \Vert_{1}\notag \\
	& = \frac{1}{2}\Vert \mc{W} + \mc{S}\Vert_{1} ,
\end{align}
where
\begin{align}
	\mc{W}&:= 
	\Tr_{B}\left[\sum_{\eta=1}^{M }\sum_{l_1,...,l_\eta =1}^{\infty}\sum_{\jmdstack{\vec{i}\in\mathbb{N}^{\eta}}{\|\vec{i}\|_{1}\leq M}}\sum_{\jmdstack{\vec{\alpha}^{j}\in\mathbf{S}^{l_{j}}}{j=1,\dots,\eta \text{ where } \vec{\mu}\neq 0}}\!\!\!\!\!\zeta^{\frac{Q+1}{2}\big(\vec{i}\cdot\vec{\mu}\big)}\times\right.\notag \\
	&\left.\mathfrak{L}_{M-i_{1}+1}^{l_{1}}(\vec{\alpha}^{1})\cdots\mathfrak{L}_{M+1-i_{1}-\dots -i_{\eta}}^{l_{\eta}}(\vec{\alpha}^{\eta})(\varrho _{SB})\right]
\end{align}
	and
\begin{align}
	& \mc{S}:= 
	\Tr_{B}\left[\varrho_{SB} + \sum_{\eta=1}^{M }\sum_{l_1,...,l_\eta =1}^{\infty}\sum_{\jmdstack{\vec{i}\in\mathbb{N}^{\eta}}{\|\vec{i}\|_{1}\leq M}}\sum_{\jmdstack{\vec{\alpha}^{j}\in\mathbf{S}^{l_{j}}}{j=1,\dots,\eta \text{ where }\vec{\mu}= 0}}\right. \notag \\
&\left.	\mathfrak{L}_{M-i_{1}+1}^{l_{1}}(\vec{\alpha}^{1})\cdots\mathfrak{L}_{M+1-i_{1}-\dots -i_{\eta}}^{l_{\eta}}(\vec{\alpha}^{\eta})(\varrho _{SB}) \right. \notag \\
&\left. - \mathcal{U}^{0}(\tau)(\varrho_{SB})\right] ,
\end{align}
with $\mathcal{U}^{0}(\tau)$ the system-bath unitary evolution superoperator generated solely by $H_{\identity}$.

\subsection{The ``Weak'' Term}
The ``strong'' term $\mc{S}$ represents the strong-measurement limit of the norm argument, i.e., the limit as $\epsilon\to\infty$.  The behavior of $\|\mc{S}\|_1$ will be analyzed in Section \ref{sec:badTerm} and shown to vanish as $M\to\infty$.  In the remainder of this section, we study the behavior of the ``weak'' term $\|\mc{W}\|_1$.  To that end, observe that for any $1\leq j\leq M$ and any $\vec{\alpha}\in\mathbf{S}^{k}$,
\bea
&\Vert \mathfrak{L}_{j}^{k}(\vec{\alpha}) \Vert \leq \Vert \mathcal{L}_{\alpha_{1}}\Vert_{\infty}\cdots \Vert \mathcal{L}_{\alpha_{k}}\Vert_{\infty} \int_{\frac{j-1}{M}\tau}^{\frac{j}{M}\tau}\int_{\frac{j-1}{M}\tau}^{t_{1}}\cdots  \\
&\qquad\int_{\frac{j-1}{M}\tau}^{t_{k-1}}\rmd t_{k} \cdots\rmd t_{2}\,\rmd t_{1} = \Vert \mathcal{L}_{\alpha_{1}}\Vert_{\infty}\cdots \Vert \mathcal{L}_{\alpha_{k}}\Vert_{\infty} \frac{(\tau/M)^{k}}{k!},\notag
\eea
where $\Vert \mathcal{L}_{\alpha_{i}}\Vert_{\infty} := \mathrm{ess\,sup}_{0\leq t\leq \tau}\|\mathcal{L}_{\alpha_{i}}(t)\|$, so that
\begin{subequations}
\begin{align}
\|\mc{W}\|_{1} & \leq \sum_{\eta=1}^{M }\sum_{l_1,...,l_\eta =1}^{\infty}\sum_{\jmdstack{\vec{i}\in\mathbb{N}^{\eta}}{\|\vec{i}\|_{1}\leq M}}\sum_{\jmdstack{\vec{\alpha}^{j}\in\mathbf{S}^{l_{j}}}{j=1,\dots,\eta \text{ where }\vec{\mu}\neq 0}}\zeta^{\frac{Q+1}{2}\big(\vec{i}\cdot\vec{\mu}\big)}\notag\\
&\Vert\mathfrak{L}_{M-i_{1}+1}^{l_{1}}(\vec{\alpha}^{1})\Vert\cdots\Vert\mathfrak{L}_{M+1-i_{1}-\dots -i_{\eta}}^{l_{\eta}}(\vec{\alpha}^{\eta})\Vert\\
& \leq \sum_{\eta=1}^{M }\sum_{\jmdstack{\vec{i}\in\mathbb{N}^{\eta}}{\|\vec{i}\|_{1}\leq M}}\sum_{\jmdstack{\vec{\mu}\in\{0,1\}^{\eta}}{\vec{\mu}\neq 0}}\zeta^{\frac{Q+1}{2}\big(\vec{i}\cdot\vec{\mu}\big)}\times\notag\\
&\sum_{l_1,...,l_\eta =1}^{\infty}\frac{(\tau/M)^{l_{1}+\dots+l_{\eta}}}{l_{1}!\cdots l_{\eta}!}\sum_{\jmdstack{g_{1},\dots,g_{\eta}\in\mathbf{S}}{\upsilon(g_{\eta}\cdots g_{j}) = \mu_{j}}}\prod_{j=1}^{\eta}\gamma_{l_{j}}(g_{j}),
\end{align}
\label{eqn:GoodTermDecomposition}
\end{subequations}
where $\gamma_{l}(g)$ is defined by 
\bea
&\sum_{\jmdstack{\vec{\alpha}\in(0,\dots,Q)^{l}}{S_{\alpha_{1}}\cdots S_{\alpha_{l}} = g}}\Vert\mathcal{L}_{\alpha_{1}}\Vert_{\infty}\cdots\Vert\mathcal{L}_{\alpha_{l}}\Vert_{\infty} \leq \gamma_{l}(g) := \notag \\
&\qquad \sum_{s=0}^{l}\binom{l}{s}J_{0}^{s}J_{1}^{l-s}f_{l-s}(g) ,
\eea
in which $f_{l}(g)$ denotes the number of ways of generating $g\in \mathbf{S}$ from $l$ non-identity group elements (i.e., $g = S_{\alpha_{1}}\cdots S_{\alpha_{l}}$).  We have used in Eq.~\eqref{eqn:GoodTermDecomposition} the triangle inequality and submultiplicativity 
\begin{align}
	\Vert AB\Vert  & \leq \Vert A\Vert \Vert B\Vert ,
\end{align}
(for any pair of operators $A$ and $B$) and the fact that if $\varrho $ is a
normalized state then $\Vert \varrho \Vert _{1}=1$.  We have also used the norm inequalities 
\begin{equation}
	\Vert AB\Vert _{1}\leq \Vert A\Vert \Vert B\Vert _{1} \; \text{ and } \;
	\Vert \mathrm{Tr}_{B}\left[ A\right] \Vert _{1} \leq \Vert A\Vert_{1},  \label{TrB}
\end{equation}
valid for any operators $A$ and $B$ acting on $\mathcal{H}_{S}\otimes \mathcal{H}_{B}$ \cite{Neumann:37, Watrous:04}.

To find a closed expression for $f_{l}(g)$, first note that $f_{l}(g)$ is constant over all $g\neq \identity$.  Then for any $g= S_{\alpha_{1}}\cdots S_{\alpha_{l}}\neq\identity$ it may be observed that $g' = S_{\alpha_{1}}\cdots S_{\alpha_{l-1}}$ can be anything other than $g$ (otherwise $S_{\alpha_{l}} = \identity$, which is forbidden) and then $g'$ and $g$ uniquely define $S_{\alpha_{l}}$.  One such choice of $g'$ is $g'=\identity$, the other $Q-1$ choices are non-identity.  The same logic applied to the case $g=\identity$ shows that all $Q$ possible choices for $g'$ are non-identity.  It follows that the quantity $f_{l}(g)$ satisfies 
\begin{subequations}
\begin{align}
	f_{l}(g) & = \sum_{g'\neq g}f_{l-1}(g') = (Q-1)f_{l-1}(g) + f_{l-1}(\identity)\\
	f_{l}(\identity) & = \sum_{g'\neq\identity}f_{l-1}(g') = Q f_{l-1}(g).
\end{align}
\end{subequations}
Therefore $f_{l}(g)$ satisfies the linear recurrence (see Appendix \ref{app:solvingLinearRecurrences})
\begin{equation}
	f_{l}(g) = (Q-1)f_{l-1}(g) + Q f_{l-2}(g).
\end{equation}
The characteristic polynomial of this recurrence, $x^{2} - (Q-1)x - Q = 0$, has roots $-1$ and $Q$, so $f_{l}(g)$ has the general form $f_{l}(g) = AQ^{l} + B(-1)^{l}$.  It is easily seen that $f_{0}(g) = 0$ and $f_{1}(g) = 1$, so $A$ and $B$ may be found by solving the linear system
\begin{subequations}
\begin{align}
	A + B & = 0\\
	AQ - B & = 1,
\end{align}
\end{subequations}
which yields
\begin{equation}A = \frac{1}{Q+1} \qquad \qquad B = -\frac{1}{Q+1}\end{equation}
so that 
\begin{equation}
	f_{l}(g) = \frac{Q^{l} - (-1)^{l}}{Q+1} \text{ for } g\neq \identity  \qquad f_{l}(\identity) = \frac{Q^{l} + Q(-1)^{l}}{Q+1}.
\end{equation}
Now, since \begin{equation}
	\sum_{s=0}^{l}\binom{l}{s}J_{0}^{s}J_{1}^{l-s}a^{l-s} = \left(J_{0} + a J_{1}\right)^{l},
\end{equation}
for any fixed $a\in\mathbb{R}$, it is readily seen that
\bes
\bea
	\gamma_{l}(g) &=& \frac{(J_{0} + Q J_{1})^{l} - (J_{0} - J_{1})^{l}}{Q+1} \text{ for } g\neq \identity \\
	 \gamma_{l}(\identity) &=& \frac{(J_{0} + Q J_{1})^{l} + Q(J_{0} - J_{1})^{l}}{Q+1} .
\eea
\ees
Defining 
\bea
\Gamma_{g}(\tau/M) := \sum_{l=1}^{\infty}\frac{(\tau/M)^{l}\gamma_{l}(g)}{l!},
\eea 
we have
\bes
\bea
\Gamma_{g}(\tau/M) &= e^{\frac{\tau J_{0}}{M}}\left[\frac{e^{\frac{\tau  Q J_{1}}{M}} - e^{-\frac{\tau J_{1}}{M}}}{Q+1}\right] \text{ for } g\neq \identity \\
 \Gamma_{\identity}(\tau/M) &= e^{\frac{\tau J_{0}}{M}}\left[\frac{e^{\frac{\tau Q J_{1}}{M}} + Qe^{-\frac{\tau J_{1}}{M}}}{Q+1}\right] - 1. 
\eea
\ees

Observe that for $g\neq \identity$,
\bea
	&\Gamma_{\identity}(\tau/M) - \Gamma_{g}(\tau/M) = e^{\frac{\tau (J_{0}-J_{1})}{M}} - 1 \notag \\
	&\qquad \begin{cases} \geq 0 \text{ for all }\tau\geq 0 & \text{when }J_{0}\geq J_{1}\\ \leq 0 \text{ for all }\tau \geq 0 & \text{when }J_{0}\leq J_{1}.\end{cases}
\eea
In addition, for a given binary vector $\vec{\mu}\in\{0,1\}^{\eta}$, the cardinality of the set of $\eta$-tuples $(g_{1},\dots,g_{\eta})$ such that $\upsilon(g_{\eta}\cdots g_{j}) = \mu_{j}$ for all $j$ is $Q^{\|\vec{\mu}\|_{1}}$, because $g_{j}$ must equal $g_{\eta}\cdots g_{j+1}$ when $\mu_{j}=0$, and $g_{j}$ may be any of the $Q$ other elements of $\mathbf{S}$ when $\mu_{j} = 1$.
Therefore, 
\begin{subequations}
\begin{align}
&\sum_{l_1,...,l_\eta =1}^{\infty}\frac{(\tau/M)^{l_{1}+\dots+l_{\eta}}}{l_{1}!\cdots l_{\eta}!}\sum_{\jmdstack{g_{1},\dots,g_{\eta}\in\mathbf{S}}{\upsilon(g_{\eta}\cdots g_{j}) = \mu_{j}}}\prod_{j=1}^{\eta}\gamma_{l_{j}}(g_{j})  =\notag\\
&\quad  \sum_{\jmdstack{g_{1},\dots,g_{\eta}\in\mathbf{S}}{\upsilon(g_{\eta}\cdots g_{j}) = \mu_{j}}}\prod_{j=1}^{\eta}\Gamma_{g_{j}}(\tau/M)\\
& \leq \begin{cases}Q^{\|\vec{\mu}\|_{1}}\Gamma_{g}(\tau/M)\big[\Gamma_{\identity}(\tau/M)\big]^{\eta-1} & J_{0}\geq J_{1}\\ Q^{\|\vec{\mu}\|_{1}}\big[\Gamma_{g}(\tau/M)\big]^{\eta} & J_{0}\leq J_{1},\end{cases}
\end{align}
\end{subequations}
since each product in the sum contains at least one $\Gamma_{g}$ ($g\neq \identity$).  Let 
\begin{align}
\label{eq:phi}
&\phi(M) := \sum_{\eta=1}^{M} \beta(\tau/M)^{\eta-1} \sum_{\jmdstack{\vec{\mu}\in\{0,1\}^{\eta}}{\vec{\mu}\neq 0}}\sum_{\jmdstack{\vec{i}\in\mathbb{N}^{\eta}}{\|\vec{i}\|_{1}\leq M}}\xi^{\vec{\mu}\cdot\vec{i}}Q^{\|\vec{\mu}\|_{1}} \\
&\qquad \xi := \zeta^{\frac{Q+1}{2}} ,\\
&\beta(\tau/M) := \\
&\begin{cases}\Gamma_{\identity}(\tau/M) = \frac{1}{Q+1}e^{\frac{\tau J_{0}}{M}}\left(e^{\frac{\tau Q J_{1}}{M}} + Qe^{-\frac{\tau J_{1}}{M}}\right) - 1 & J_{0} \geq J_{1}\notag\\
\Gamma_{g}(\tau/M) = \frac{1}{Q+1}e^{\frac{\tau J_{0}}{M}}\left(e^{\frac{\tau Q J_{1}}{M}} - e^{-\frac{\tau J_{1}}{M}}\right) & J_{0} \leq J_{1}.\end{cases}
\end{align}
Therefore we have
\begin{align}
&\|\mc{W}\|_1\leq \Gamma_{g}(\tau/M)\phi(M) .
\label{eqn:startingPoint}
\end{align}

A closed form will now be derived for $\Gamma_{g}\phi$ and shown to converge to zero as $M^{-1}$ for arbitrary fixed $\tau>0$ and $\xi\in(0,1)$ as $M\to\infty$.  To that end, let $r$ represent the number of non-zero elements in the $\vec{\mu}$ vector, and $u$ represent $\vec{\mu}\cdot \vec{i}$.  Then
\begin{align}
&\sum_{\jmdstack{\vec{\mu}\in\{0,1\}^{\eta}}{\vec{\mu}\neq 0}}\sum_{\jmdstack{\vec{i}\in\mathbb{N}^{\eta}}{\|\vec{i}\|_{1}\leq M}}\xi^{\vec{\mu}\cdot\vec{i}}Q^{\|\vec{\mu}\|_{1}}= \notag \\
&\quad  \sum_{r=1}^{\eta}Q^{r}\binom{\eta}{r}\sum_{u=r}^{M-(\eta-r)}\xi^{u}\#\{\textrm{ordered } r\textrm{-partitions of } u\}\times\notag\\
&\quad\#\{\textrm{orderered } (\eta-r)\textrm{-partitions of } M-u \textrm{ or less}\},
\end{align}
where an ordered $k$-partition of $n$ is an ordered set of $k$ positive integers that sum to $n$: $j_{1}+\dots + j_{k} = n$.  The number of such partitions is $\binom{n-1}{k-1}$, as may be seen by casting the question as an occupancy problem \cite{Feller:68} by considering the placing of $k-1$ physical separators between a linear arrangement of $n$ physical objects (see Fig.~\ref{fig:linearPartition}).
\begin{figure}[ht]
\centering
\leavevmode
\xymatrix@=5pt@M=5pt{& & & \ar@{-}[dd] & & \ar@{-}[dd] & &  & &  & \\\bullet & & \bullet & &\bullet & &\bullet & &\bullet & & \bullet\\ & & & & & & & & & & }
\caption{Example partition of six elements into three sets by choosing the positions of two separators from among the five possible gaps between adjacent elements.}
\label{fig:linearPartition}
\end{figure}
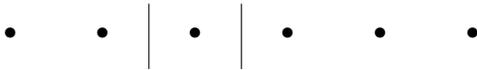
Moreover, by the same reasoning, the number of ordered $k$-partitions of $n$ or less is $\binom{n}{k}$, seen by considering the placing of $k$ separators between $n+1$ objects, yielding $j_{1}+\dots+j_{k+1} = n+1$, and then discarding $j_{k+1}$.  So 
\begin{align}
	&\sum_{\jmdstack{\vec{\mu}\in\{0,1\}^{\eta}}{\vec{\mu}\neq 0}}\sum_{\jmdstack{\vec{i}\in\mathbb{N}^{\eta}}{\|\vec{i}\|_{1}\leq M}}\xi^{\vec{\mu}\cdot\vec{i}} Q^{\|\vec{\mu}\|_{1}}= \notag \\
&\qquad \sum_{r=1}^{\eta}Q^{r}\binom{\eta}{r}\sum_{u=r}^{M-(\eta-r)}\xi^{u}\binom{u-1}{r-1}\binom{M-u}{\eta-r},
\end{align}
and we find that $\phi(M)$ in Eq.~\eqref{eq:phi} is equal to 
\begin{subequations}
\label{eqn:rawmultisums}
\begin{align}
	\phi(M) & = \sum_{\eta=1}^{M} \beta(\tau/M)^{\eta-1}\sum_{r=1}^{\eta}Q^{r}\binom{\eta}{r}\times\notag \\
	&\qquad \sum_{u=r}^{M-(\eta-r)}\xi^{u}\binom{u-1}{r-1}\binom{M-u}{\eta-r}\label{eqn:etarusum}\\
	& = \sum_{\eta,u,r=1}^{M} \beta(\tau/M)^{\eta-1}\xi^{u}Q^{r}\binom{\eta}{r}\binom{u-1}{r-1}\binom{M-u}{\eta-r},\label{eqn:uniformEtaruSum}
\end{align}
\end{subequations}
where the limits have been extended to $1,\dots, M$ since the additional terms are all zeros.

\subsection{Linear Recurrences}
Let $\Phi$ be the summand of Eq.~\eqref{eqn:rawmultisums}, i.e., \begin{equation}
\Phi(M,u,\eta,r):=\beta^{\eta-1} \xi^{u} Q^{r} \binom{\eta}{r}\binom{u-1}{r-1}\binom{M-u}{\eta-r},
\end{equation} where we regard $\beta$ as a fixed quantity for the moment (its $M$ dependence will be reintroduced later).  It may be observed that $\Phi$ is hypergeometric in all four variables, i.e., $\Phi(M+1,u,\eta,r)/\Phi(M,u,\eta,r)$ is a rational function of $M$, $u$, $\eta$, and $r$; likewise for the corresponding shift ratios on $u$, $\eta$, and $r$.  As a result, $\Phi$ admits a linear recurrence relation such that the coefficients are independent of $u$, $\eta$, and $r$, and the recurrence may be found algorithmically \cite{Petkovsek1996}.  Wegschaider's MultiSum package \cite{Wegschaider1997, Lyons2002} for Mathematica,  implementing a version of Sister Celine's method \cite{Petkovsek1996}, yielded the following telescoping linear recurrence for this summand:
\begin{align}
&\xi (1 + \beta + Q \beta)  \Phi(M-2,u,\eta,r)-\notag \\
& \big(1 + \beta + (1 + Q\beta)\xi\big) \Phi(M-1,u,\eta,r)+\Phi(M,u,\eta,r)\nonumber\\
    & = \Delta_{u}\Big[\beta \Phi(M-1,u,\eta,r+1)+\Phi(M-1,u,\eta+1,r+1) \notag \\
    &\quad -\Phi(M,u,\eta+1,r+1)\Big]\nonumber\\
    & \quad + \Delta_{\eta}\Big[-\xi \Phi(M-2,u,\eta,r+1) + (1 + \xi)\times \notag \\
    & \qquad \Phi(M-1,u,\eta,r+1)-\Phi(M,u,\eta,r+1)\Big]\nonumber\\
    & \quad + \Delta_{r}\Big[-(1+\beta) \xi \Phi(M-2,u,\eta,r)+(1 + \beta+\xi) \times \notag \\
    &\qquad \Phi(M-1,u,\eta,r)-\Phi(M,u,\eta,r)\Big], \label{eqn:summandRecurrence}
\end{align}

where $\Delta_{\eta}$ is the \textit{forward shift difference} operator on the variable $\eta$: $\Delta_{\eta}X(M,u,\eta,r) = X(M,u,\eta+1,r) - X(M,u,\eta,r)$ for any expression $X$; likewise for $\Delta_{r}$ and $\Delta_{u}$.  The recurrence \eqref{eqn:summandRecurrence} may be verified by expanding the $\Delta$ expressions and collecting terms, leading to the equivalent equation
\begin{align}
0 & = Q \beta \xi \Phi(M-2, u, \eta, r)  + \beta\xi \Phi(M-2, u, \eta, r+1)  \notag \\
&+\xi \Phi(M-2, u, \eta+1, r+1) - Q \beta\xi \Phi(M-1, u, \eta, r)\nonumber\\
& - \xi \Phi(M-1, u, \eta+1, r+1) \notag \\
&- \beta \Phi(M-1, u+1, \eta, r+1) \notag \\
&- \Phi(M-1, u+1, \eta+1, r+1)\nonumber\\
&  + \Phi(M, u+1, \eta+1, r+1).
\label{eqn:simplifiedSummandRecurrence}
\end{align}
Dividing Eq.~\eqref{eqn:simplifiedSummandRecurrence} by $Q\Phi(M-2,u,\eta,r)$ and multiplying by the common denominator $r(r+1)(M-u-\eta + r - 1)$, Eq.~\eqref{eqn:summandRecurrence} is then found to be equivalent to
\begin{align}
0 & = \beta \xi r(r+1)(M-u-\eta+r-1) + \beta\xi (\eta-r)^{2}(u-r) \notag \\
&+ \beta\xi (\eta+1)(u-r)(M-u-\eta+r-1) \nonumber\\
& - \beta\xi [r(r+1)+ (\eta+1)(u-r)](M-u-1) \notag\\
& - \beta \xi(\eta-r)^{2}u  - \beta\xi (\eta+1)u(M-u-\eta+r-1) \nonumber\\
& + \beta\xi (\eta+1)u(M-u-1) ,
\end{align}
which is easily verified to be true.

Now we wish to turn this recurrence on the summand $\Phi(M,u,\eta,r)$ into a recurrence on the sum $\phi(M)$.  Doing this requires summing both sides of Eq.~\eqref{eqn:summandRecurrence} and using the fact that the right hand side leads to telescoping sums: $\sum_{k=1}^{n}\Delta_{k}X(k) = X(n+1)-X(1)$.  This computation, carried out in Appendix \ref{app:linearRecurrenceInhom}, results in the inhomogeneous linear recurrence relation (with constant coefficients) 
\begin{align}
&(1 + \beta+ Q\beta) \xi \phi(M-2) - \big(1 + \beta +(1+Q\beta)\xi\big) \phi(M-1) \notag\\
&\qquad + \phi(M) =  Q \beta(1+\beta)^{M-2}\xi.\label{eqn:sumRecurrence}
\end{align}
This recurrence for the sum has the characteristic polynomial \begin{equation}
x^{2} - \big(1+\beta + (1+Q\beta)\xi\big)x + (1+\beta + Q\beta)\xi = 0
\end{equation} with roots 
\begin{subequations}
\begin{align}
\gamma_{\pm} & := \frac{\big(1+\beta +(1+Q\beta)\xi\big)}{2} \notag\\
& \pm \frac{\sqrt{(1+\beta + (1+Q\beta)\xi\big)^{2} - 4(1+\beta+Q\beta)\xi}}{2}\\
& = \frac{\big(1+\beta + (1+Q\beta)\xi\big)}{2} \notag \\
&\pm \frac{\sqrt{\big(1+\beta-(1+Q\beta)\xi\big)^{2} + 4 Q \beta^{2}\xi}}{2}.
\end{align}
\end{subequations}
Furthermore, by plugging a sequence of the form $\phi(M)=a(1+\beta)^{M}$ into Eq.~\eqref{eqn:sumRecurrence}, it may be seen that $-\frac{(1+\beta)^{M}}{\beta}$ is a particular solution to the inhomogeneous recurrence.  Thus Eq.~\eqref{eqn:sumRecurrence} admits the general solution 
\begin{equation}\phi(M) = A_{+}\gamma_{+}^{M-1} + A_{-}\gamma_{-}^{M-1} - \frac{(1+\beta)^{M}}{\beta}\label{eqn:generalSolution}\end{equation}
for some coefficients $A_{+}$ and $A_{-}$.  The initial values of $\phi(M)$ may be worked out to be
\begin{subequations}
\begin{align}
	\phi(1) & = Q \xi\\
	\phi(2) & = Q\xi\big(1+2\beta+(1+Q\beta)\xi\big).
\end{align}
\end{subequations}
Thus, we seek $A_{+}$ and $A_{-}$ such that $A_{+}+A_{-} = Q\xi + (1+\beta)/\beta$ and $A_{+}\gamma_{+} + A_{-}\gamma_{-} = Q\xi\big(1+2\beta + (1+Q\beta)\xi\big) + (1+\beta)^{2}/\beta$.  Solving this linear system for $A_{+}$ and $A_{-}$ yields
\begin{equation}A_{\pm} = \frac{Q\beta\xi(\gamma_{\pm}+\beta) + (1+\beta)\big[(1+\beta)-\gamma_{\mp}\big]}{\beta(\gamma_{\pm}-\gamma_{\mp})}.\end{equation}
With these values for $A_{+}$ and $A_{-}$, and now regarding $\beta$, $\gamma_{\pm}$, and $A_{\pm}$ all as functions of $M$, Eq.~\eqref{eqn:generalSolution} is an exact closed expression for the sum $\phi(M)$.  So we finally conclude that
\begin{align}
	\|\mc{W}\|_{1}&\leq \frac{\Gamma_{g}(M)}{\beta(M)}\Big[\beta(M) A_{+}(M)\gamma_{+}^{M-1}(M) + \notag \\
	&\beta(M) A_{-}(M)\gamma_{-}^{M-1}(M) - (1+\beta(M))^{M}\Big],
	\label{eq:82}
\end{align}
where \begin{align}
\beta(M) = \begin{cases}\Gamma_{\identity}(M) & J_{0}\geq J_{1}\\\Gamma_{g}(M) & J_{0}\leq J_{1}.\end{cases}\end{align}


\subsection{The ``Strong'' Term}
\label{sec:badTerm}
Following arguments similar to those elaborated thus far, we may now readily compute an upper bound for the ``strong'' term $\|\mc{S}\|_{1}$.  Explicitly,	
\begin{subequations}
\begin{align}
\Vert \mc{S}\Vert_{1} & = \Big\Vert \Tr_{B}\sum_{\eta=1}^{M}\sum_{l_1,...,l_\eta =1}^{\infty}\sum_{\jmdstack{\vec{i}\in\mathbb{N}^{\eta}}{\|\vec{i}\|_{1}\leq M}}\sum_{\jmdstack{\vec{\alpha}^{j}\in\mathbf{S}^{l_{j}}}{\vec{\mu}= 0 \text{ but not all } \alpha_{k}^{j} = \identity}}\!\!\!\!\!\!\!\!\!  \notag \\
&  \qquad \mathfrak{L}_{M-i_{1}+1}^{l_{1}}(\vec{\alpha}^{1})\cdots\mathfrak{L}_{M+1-i_{1}-\dots -i_{\eta}}^{l_{\eta}}(\vec{\alpha}^{\eta})(\varrho _{SB})\Big\Vert_{1}\\
& \leq \sum_{\eta=1}^{M}\sum_{l_1,...,l_\eta =1}^{\infty}\sum_{\jmdstack{\vec{i}\in\mathbb{N}^{\eta}}{\|\vec{i}\|_{1}\leq M}}\sum_{\jmdstack{\vec{\alpha}^{j}\in\mathbf{S}^{l_{j}}}{\vec{\mu}= 0 \text{ but not all } \alpha_{k}^{j} = \identity}}\!\!\!\!\!\!\notag \\& \qquad \big\Vert\mathfrak{L}_{M-i_{1}+1}^{l_{1}}(\vec{\alpha}^{1})\big\Vert\cdots\big\Vert\mathfrak{L}_{M+1-i_{1}-\dots -i_{\eta}}^{l_{\eta}}(\vec{\alpha}^{\eta})\big\Vert\\
& \leq \sum_{\eta=1}^{M}\sum_{l_1,...,l_\eta =1}^{\infty}\sum_{\jmdstack{\vec{i}\in\mathbb{N}^{\eta}}{\|\vec{i}\|_{1}\leq M}}\times \notag \\
&\qquad \frac{(\tau/M)^{l_{1}+\dots +l_{\eta}}}{l_{1}!\cdots l_{\eta}!}\left[\prod_{j=1}^{\eta}\gamma_{l_{j}}(\identity) - J_{0}^{l_{1}+\dots +l_{\eta}}\right]\\
& = \sum_{\eta=1}^{M}\binom{M}{\eta}\left[\Gamma_{\identity}^{\eta}(\tau/M) - \left(e^{\tau J_{0}/M} - 1\right)^{\eta}\right]  \\
\label{eq:84e}
& = [1+\Gamma_{\identity}(\tau/M)\big]^{M} - e^{\tau J_{0}}\\
& = e^{\tau J_{0}}\left[\left(\frac{e^{\frac{Q \tau J_{1}}{M}} + Qe^{-\frac{\tau J_{1}}{M}}}{Q+1}\right)^{M} - 1\right].
\end{align}
\end{subequations}


\subsection{Total Distance Upper Bound}
\label{sec:totalUpperBound}
Combining the results for the ``weak'' and ''strong'' terms, and recalling that $A_{\pm}(M)$ and $\gamma_{\pm}(M)$ depend on the form of $\beta(M)$
\begin{equation}\beta(M) = \begin{cases}\Gamma_{\identity}(M) = e^{\frac{\tau J_{0}}{M}}\left[\frac{e^{\frac{\tau Q J_{1}}{M}} + Qe^{-\frac{\tau J_{1}}{M}}}{Q+1}\right] - 1 & J_{0}\geq J_{1}\\ & \\ \Gamma_{g}(M) = e^{\frac{\tau J_{0}}{M}}\left[\frac{e^{\frac{\tau  Q J_{1}}{M}} - e^{-\frac{\tau J_{1}}{M}}}{Q+1}\right] & J_{0}\leq J_{1},\end{cases}
\end{equation}
with $q = |\mathbf{S}|/2 = (Q+1)/2$, $J_{m} = \max\{J_{0},J_{1}\}$, and 
\begin{equation}	
\Gamma_{+} := \begin{cases}\Gamma_{\identity}(M) & J_{0}\geq J_{1}\\\Gamma_{g}(M) & J_{0}\leq J_{1}\end{cases} \quad \Gamma_{-}:=\begin{cases}\Gamma_{\identity}(g) & J_{0}\geq J_{1}\\\Gamma_{\identity}(M) & J_{0}\leq J_{1},
\end{cases}
\end{equation}
we obtain, using Eqs.~\eqref{eq:82} and \eqref{eq:84e}, the upper bound on the distance metric
\begin{subequations}
\begin{align}
D&[\varrho_{S}(\tau), \varrho_{S}^{0}(\tau)] \leq 
\frac{1}{2}\big(\|\mc{S}\|_{1}+\|\mc{W}\|_{1}\big)\nonumber\\
	& \leq [1+\Gamma_{\identity}(M)\big]^{M} - \frac{\Gamma_{-}(M)}{\Gamma_{\identity}(M)}[1+\Gamma_{+}(M)\big]^{M}- e^{\tau J_{0}} \notag \\
	&+ \Gamma_{g}(M)A_{+}(M)\gamma_{+}^{M-1}(M) + \Gamma_{g}(M)A_{-}(M)\gamma_{-}^{M-1}(M)\label{eqn:fullBound2}\\
	& = \sum_{j=1}^\infty \textrm{B}_j \frac{1}{M^j} ,
\label{eqn:totalAsymptoticBound}
\end{align}
\end{subequations}
where
\bea
\textrm{B}_1 = 
\left[e^{\tau J_{0}}\left(\frac{Q\tau^{2} J_{1}^{2}}{4}\right) + e^{\tau J_{m}}\frac{Q\tau J_{1}}{2}\left(1 + \tau J_{m}\right)\frac{\zeta^{q}}{1-\zeta^{q}}\right] \notag \\
\label{eq:B_1}
\eea
The detailed asymptotic analysis leading to Eqs. \eqref{eqn:totalAsymptoticBound} and \eqref{eq:B_1} is carried out in Appendix \ref{app:asymptoticAnalysis}.  This completes the proof of Theorem \ref{Th:1}.
It may be observed that the rate of convergence (i.e., the coefficient of $M^{-1}$) blows up as $\tau \to \infty$ and also as $\zeta \to 1$ (i.e., the weak measurement limit as $\epsilon \to 0$).  

In the next section we analyze the conditions under which the distance bound converges to zero.

\section{Parameter Trade-Offs}
\label{sec:parameterTradeOffs}

\subsection{Trade-off between $\tau$ and $M$}

The asymptotic behavior of the distance upper bound, presented in Section \ref{sec:totalUpperBound} and derived in Appendix \ref{app:asymptoticAnalysis}, shows that the bound converges to zero as $1/M$ when the time $\tau$ is held fixed.  Here we consider how the number of measurements needed for convergence depends upon the final time $\tau$.  In order to address this question, we consider the convergence of the bound when $\tau$ is allowed to increase with $M$.  If $\tau$ increases too quickly, the distance upper bound will no longer converge, so we seek the fastest growing function $f(M)$ such that when $\tau = f(M)$, the bound still converges to zero.  

\begin{mylemma}
Allowing only $\tau$ and $M$ to vary, the distance upper bound~\eqref{eqn:totalAsymptoticBound} converges to zero as $M\rightarrow \infty$ provided 
\begin{align}
\label{eq:conv-suff}
\tau J_{0} = a \log(M)\quad 
	\begin{cases}
	 a<1  \quad &\textrm{if}\ \ J_0\geq J_1 \\
	 a< \frac{J_0}{J_1}  \quad &\textrm{if}\ \ J_0\leq J_1
	\end{cases} .
\end{align}
\end{mylemma}

\begin{proof}
Suppose that $\tau J_{0} = a \log(M)$ for some $a>0$.  Then defining $\lambda = J_{1}/J_{0}$, 
\begin{subequations}
\begin{align}
	\Gamma_{\identity}(M) & = e^{\frac{a\log(M)}{M}}\left(\frac{e^{\frac{Q \lambda a \log(M)}{M}} + Qe^{-\frac{\lambda a\log (M)}{M}}}{Q+1}\right) - 1  \notag \\
	&= \frac{a\log(M)}{M} + \frac{\big(1 + Q \lambda^{2}\big)a^{2}\log^{2}(M)}{2M^{2}} \notag \\
	&+ O\left(\frac{\log^{3}(M)}{M^{3}}\right)\\
	\Gamma_{g}(M) & = e^{\frac{a\log(M)}{M}}\left(\frac{e^{\frac{Q\lambda a\log(M)}{M}} - e^{-\frac{\lambda a\log(M)}{M}}}{Q+1}\right) \notag \\
	&= \frac{\lambda a\log(M)}{M} + \frac{\big(2\lambda + (Q-1)\lambda^{2}\big)a^{2}\log^{2}(M)}{2M^{2}} \notag \\
	&+ O\left(\frac{\log^{3}(M)}{M^{3}}\right).
\end{align}
\end{subequations}
Repeating the analysis of Appendix \ref{app:asymptoticAnalysis} with the general form \begin{equation}
\beta(M) = X\frac{\log(M)}{M} + Y\frac{\log^{2}(M)}{M^{2}} + O\left(\frac{\log^{3}(M)}{M^{3}}\right)
\end{equation} 
yields
\begin{align}
\label{eq:Dbound}
	&D[\varrho_{S}(\tau), \varrho_{S}^{0}(\tau)] \leq \notag \\
	&\ \ 
	\begin{cases}
	\left[\frac{Q \lambda^{2}a^{2}\log^{2}(M)}{4} + \frac{Q}{2}\big(\lambda a \log(M) + \right. \\
	\left. \quad \lambda a^{2}\log^{2}(M)\big)\frac{\zeta^{\frac{Q+1}{2}}}{1-\zeta^{\frac{Q+1}{2}}}\right]\frac{1}{M^{1-a}} + O\left(\frac{\log^{3}(M)}{M^{2-a}}\right) \\
	\textrm{if}\ \ J_{0} \geq J_{1}\\
	&\\
\left(\frac{\lambda^{2}a^{2}\log^{2}(M)}{4M^{1-a}}\right) + \left[\frac{Q}{2}\left(\lambda a \log(M) + \right.\right. \\
\left.\left. \quad  \lambda^{2}a^{2}\log^{2}(M)\right)\frac{\zeta^{\frac{Q+1}{2}}}{1-\zeta^{\frac{Q+1}{2}}}\right]\frac{1}{M^{1-\lambda a}} + O\left(\frac{\log^{3}(M)}{M^{2-\lambda a}}\right) \\
 \textrm{if}\ \ J_{0} \leq J_{1}.
 \end{cases}
\end{align}
Because of the $M^{1-a}$ in the denominator of the leading order terms, the bound for $J_{0}\geq J_{1}$ converges to zero as $M\to \infty$ if $a<1$ and diverges if $a\geq 1$.  In the case $J_{0}\leq J_{1}$, where $\lambda > 1$, the bound converges if $a<1/\lambda$ and diverges if $a\geq 1/\lambda$.   
\end{proof}

The closer $a$ is to the cutoff value of either $1$ or $1/\lambda$, the slower the convergence of the bound becomes. We can conclude that Eq.~\eqref{eq:conv-suff} is a sufficient condition for convergence of the upper bound. Divergence of the upper bound does not preclude the distance $D$ from converging to zero in the limit of a large number of measurements, and additional analysis is required to settle whether this is possible. The difference between the $J_0\leq J_1$ and $J_0\geq J_1$ cases can be understood intuitively in terms of the relative importance of the ``perturbation" (coupling of the system to the bath, $J_1$) and the ``unperturbed" evolution ($J_0$). It is to be expected that protection via the Zeno effect should be more effective when the perturbation is ``weak", and as we have seen, indeed the cutoff value of $a$ is larger in this case.

\subsection{Trade-off between $\epsilon$ and $M$}

Another important trade-off is between the strength of the measurements $\epsilon$ and the number of measurements required for effective convergence.  

\begin{mylemma}
Allowing only $\epsilon$ and $M$ to vary, the distance upper bound~\eqref{eqn:totalAsymptoticBound} converges to zero as $M\rightarrow \infty$ provided 
\bea
\label{eq:eps-prop}
\epsilon > M^{a-1/2}
\eea
for some $a>0$.
\end{mylemma}
\begin{proof}
Consider the large $M$, small $\epsilon$	 regime, where the bound \eqref{eqn:totalAsymptoticBound} is approximately $\mathrm{B}_1 /M$.
Then, since $\zeta^{q}/(1-\zeta^{q})$ is large when $\epsilon$ is small [recall that $\zeta = 1/\cosh(\epsilon)$], we can neglect the constant first term in Eq.~\eqref{eq:B_1},
\begin{align}
e^{\tau J_{m}}\frac{Q \tau J_{1}}{2\mathrm{B}_1}\left(1 + \tau J_{m}\right)&\approx \frac{1-\zeta^{q}}{\zeta^{q}}=\frac{1}{\zeta^{q}} - 1\notag \\
&\approx q(1-\zeta) ,
\end{align}
so that 
\begin{align}
\zeta \approx 1-e^{\tau J_{m}}\frac{Q \tau J_{1}}{2q\mathrm{B}_1}\left(1 + \tau J_{m}\right),
\end{align}
and therefore, since $\zeta\approx 1-\epsilon^{2}/2$ in this regime,
\begin{align}
\mathrm{B}_1 \frac{1}{M} & \approx \frac{1}{\epsilon^{2}M}e^{\tau J_{m}}{Q \tau J_{1}}\left(1 + \tau J_{m}\right)/q \notag \\
& < M^{-2a},
\end{align}
where in the last line we used Eq.~\eqref{eq:eps-prop}.  
\end{proof}
This quantifies the previously mentioned trade-off that the weaker the measurements, the greater $M$ must be to compensate, i.e., the slower the convergence rate as $M\to\infty$. Conversely, we can interpret Eq.~\eqref{eq:eps-prop} as saying that the measurement strength may not decline faster than the inverse square root of the number of measurements.

\subsection{Fixed Nonzero Measurement Interval}

As an alternative analysis of the time-scaling issue, fix some $\Delta \tau>0$ such that the measurements are separated by an interval $\Delta \tau$, so that $\tau = M\Delta\tau$.  Under these conditions, what value of $M$ minimizes the distance upper bound?  Let $\delta = J_{0}\Delta \tau $ and $\lambda = J_{1}/J_{0}$ and consider first the strong-measurement limit:
\begin{equation}\|\mc{S}\|_{1}\leq f(M):=  e^{\delta M}\left[\left(\frac{e^{Q\lambda \delta} + Qe^{-\lambda \delta}}{Q+1}\right)^{M} -1 \right].\end{equation}
Then, taking $M$ to be a continuous variable for a moment, \begin{equation}\frac{\rmd f}{\rmd M} = \delta f(M) + \log\left(\frac{e^{Q\lambda \delta} + Qe^{-\lambda \delta}}{Q+1}\right) \big(f(M)+e^{\delta M}\big).\label{eqn:derivStrongBound}\end{equation}
Observe that since $Q = 2^{\bar{Q}}-1 \geq 1$, 
\bea
\frac{e^{Q\lambda \delta} + Qe^{-\lambda \delta}}{Q+1} &=& 1 + \sum_{k=2}^{\infty}\frac{Q\big((-1)^{k}+Q^{k-1}\big)}{Q+1}\frac{(\tau J_{1}/M)^{k}}{k!} \notag \\
&>& 1,
\eea
the logarithm in Eq.~\eqref{eqn:derivStrongBound} is positive, as are all the other terms in \eqref{eqn:derivStrongBound}, so $\rmd f/\rmd M > 0$ for all $M>0$.  Therefore on the domain of positive integers $M$, $f(M)$ is minimized at $M=1$.  Thus, if laboratory conditions dictate a minimum interval between measurements, the upper bound on the distance indicates that the best strategy for minimizing the distance is \textit{not} to let $M\to \infty$ seeking eventual convergence (which will never come under these assumptions), but rather to endure only one such interval, terminating with a single measurement event.  The resulting bound (for $M=1$) is then
\begin{align}
\|\mc{S}\|_{1}&\leq f(1) =  e^{\delta}\left[\left(\frac{e^{Q\lambda \delta} + Qe^{-\lambda \delta}}{Q+1}\right) -1 \right] \notag \\
&= \frac{Q\lambda^{2}}{2}(\delta^{2} + \delta^{3}) + O(\delta^{4})  \notag \\
& {= \frac{Q}{2} (J_1 \Delta\tau)^2 (1+J_0\Delta\tau) +O[(J_0\Delta\tau)^{4}]}
\label{eq:M=1}
\end{align}
so that minimization of the bound can only be accomplished by minimizing either $\Delta \tau$ or $J_1$.  

This result does not mean that the Zeno effect fails to provide protection beyond $M=1$, but rather that (our bound on) the protection quality gradually declines as the elapsed time grows. The relevant yardstick for comparison is then the distance between the ideal and actual state in the absence of any protection. This distance is easily estimated using first order perturbation theory (the Dyson series) to be $O(J_1\Delta\tau)$. Since Eq.~\eqref{eq:M=1} shows that a single measurement already modifies the distance to $O[(J_1\Delta\tau)^2]$, protection is achieved provided $J_1\Delta\tau<1$. Subsequent measurements, or longer evolution times in the case without measurement, modify these estimates to $O[(J_1M\Delta\tau)^2]$ and $O(J_1 M\Delta\tau)$, respectively, so that the conclusion about the possibility of an advantage from measurements with finite and fixed $\Delta\tau$ are unchanged.

\section{Conclusion}
\label{sec:conclusion}

Two protocols have been presented for the protection from the environment of an arbitrary, unknown state encoded in some stabilizer quantum error correction (or detection) code.  These protocols involved frequent weak non-selective measurement of either all elements of the stabilizer group of the code, or of just a minimal generating set.  Rigorous upper bounds were obtained on the distance between the final state under these protocols and the idealized final state in the absence of any interaction with the environment.  The bounds demonstrate that the protocols exhibit the desired protection in the limit of many measurement cycles.  Moreover, the bounds offer information about the degree of protection attainable with finite resources (e.g., finitely many measurement cycles), as well as trade-offs among the various relevant physical parameters.

Future research in this area may proceed along different lines.  First, while the protocols based on non-selective measurements presented herein are well-suited to protecting ensembles from the environment, they may not be ideal for protecting individual quantum systems.  Further investigation is needed to determine whether (and in what sense) the corresponding protocols based on weak \textit{selective} measurements realize a Zeno effect resulting in protection from the environment.  A result in this direction would more directly relate the Zeno effect to traditional quantum error correction.  Additionally, it may be interesting and fruitful to expand the class of error correcting codes on which the protocols are based, perhaps to include non-abelian codes.  For example, a Zeno protocol based on the Bacon-Shor code \cite{Bacon:05} may have advantages in that the measurements need be only 2-local.


\section{Acknowledgment} 
This research was supported by the ARO MURI grant W911NF-11-1-0268, by the Department of Defense, by the Intelligence Advanced Research Projects Activity (IARPA)
via Department of Interior National Business Center contract number
D11PC20165, and by
NSF grants No. CHE-924318 and CHE-1037992. 
The U.S. Government is authorized to reproduce and distribute
reprints for Governmental purposes notwithstanding any copyright annotation
thereon. The views and conclusions contained herein are those of
the authors and should not be interpreted as necessarily representing the
official policies or endorsements, either expressed or implied, of IARPA,
DoI/NBC, or the U.S. Government.

\bibliographystyle{apsrev4-1}
\bibliography{WMQZE_Refs}

\appendix

\section{Isotypical Decompositions Induced by the Stabilizer Group}
\label{app:isotypicalDecompositions}
\begin{mylemma}
Let $\HermOps$ denote the space of Hermitian operators on $\mathcal{H} = \mathcal{H}_{S}\otimes \mathcal{H}_{B}$.  
  The stabilizer group acts on $\HermOps$ by conjugation [i.e. the adjoint action $\Ad_{S}(A)\mapsto S A S$], which describes a representation of $\mathbf{S}$ ($\simeq \mathbb{Z}_{2}^{\bar{Q}}$).  There is then a unique isotypical decomposition of $\HermOps$ into subspaces $W_{g}$: \begin{equation}\HermOps = \bigoplus_{g\in\mathbf{S}}W_{g} = \bigoplus_{g\in\mathbf{S}}\hat{W}_{g}^{\oplus b_{g}},\end{equation} where each $W_{g}$ is an invariant subspace of the representation, comprising $b_{g}$ copies of the one-dimensional irreducible representation $\hat{W}_{g}$, such that for any Hermitian matrix $A \in W_{g}$, $SAS = (-1)^{\sigma_{g}(S)}A $.  Since each $\hat{W}_{g}$ is one-dimensional (because $\mathbf{S}\simeq \mathbb{Z}_{2}^{\bar{Q}}$ is abelian), the exponent $b_{g}$ is the dimension of the subspace $W_{g}$.  With $\Tr(\identity) = 2^{n}$ and all other elements of $\mathbf{S}$ traceless, $b_{g} = 2^{2n-\bar{Q}}(\dim(\mathcal{H}_{B}))^{2} = 2^{2k+\bar{Q}}(\dim(\mathcal{H}_{B}))^{2}$ for all $g$.  The $2^{2k+\bar{Q}} (\dim(\mathcal{H}_{B}))^{2}$-dimensional subspaces $W_{g}$ are all orthogonal under any $\Ad(\mathbf{S}\otimes \identity_{B})$-invariant inner product on $W$.  
\end{mylemma}
\begin{proof}
By \cite[Corollary 2.16]{Fulton1991}, 
\bea
b_{g} &=& \frac{1}{2^{\bar{Q}}}\sum_{S\in\mathbf{S}}\Tr(\Ad_{S})(-1)^{\sigma_{g}(S)} = 2^{2n-\bar{Q}}(\dim(\mathcal{H}_{B}))^{2}  \notag \\
&=& 2^{2k+\bar{Q}}(\dim(\mathcal{H}_{B}))^{2}
\eea
under the assumption that only $S = \identity$ has non-zero trace and $\Tr(\identity) = \dim(\mathcal{H}) = 2^{n}\dim(\mathcal{H}_{B})$, so that $\Tr(\Ad_{\identity}) = [\Tr(\identity)]^{2} = 2^{2n}\dim(\mathcal{H}_{B})^{2}$.

If $A_{g}\in W_{g}$ and $A_{h}\in W_{h}$, then, using the $\Ad(\mathbf{S})$-invariance of the inner product (invariance with respect to conjugation by $\mathbf{S}$), we get
\bea
\langle A_{g},A_{h}\rangle &=& \langle SA_{g}S, S A_{h}S\rangle = (-1)^{\sigma_{g}(S) + \sigma_{h}(S)}\langle A_{g}, A_{h}\rangle \notag \\
&=& (-1)^{\sigma_{gh}(S)}\langle A_{g}, A_{h}\rangle
\eea
for any $S\in\mathbf{S}$.  If $g\neq h$, Lemma \ref{lem:homomorphismProperties} shows that there exists an $S\in\mathbf{S}$ such that $\sigma_{gh}(S) = 1$.  It must therefore hold that $\langle A_{g}, A_{h}\rangle = 0$, so $W_{g}$ and $W_{h}$ are orthogonal subspaces.
\end{proof}


\section{Solving Linear Recurrences}
\label{app:solvingLinearRecurrences}
In this appendix, we review a simple method for solving a class of second order inhomogeneous linear recurrences with constant coefficients.  Let $f_{k}-a_{1}f_{k-1}-a_{2}f_{k-2} = g_{k}$ be the linear recurrence with inhomogeneity $g_{k} = b c^{k}$ for real numbers $a_{1}$, $a_{2}$, $b$, and $c$, and given initial conditions $f_{0}$ and $f_{1}$.  Define the characteristic polynomial to be $p(x) = x^{2}-a_{1}x-a_{2}$.  The class of problems we will consider are those for which the characteristic polynomial admits two distinct roots, i.e., where $a_{1}^{2} + 4a_{2} > 0$, and where $c$ is different from both roots.  As in the theory of linear differential equations, we solve such problems by finding all solutions to the homogeneous problem $f_{k}-a_{1}f_{k-1}-a_{2}f_{k-2} = 0$, then seeking a particular solution for the inhomogeneous problem, and finally identifying within the resultant affine space of solutions, the unique solution that satisfies the initial conditions.  To that end, let 
\begin{equation}
	\vec{w}^{k} = \begin{bmatrix}f_{k+1}\\f_{k}\end{bmatrix},\quad \text{ and } \quad A = \begin{bmatrix}a_{1} & a_{2}\\1 & 0\end{bmatrix}.
\end{equation}
Then the homogeneous linear recurrence may be written $\vec{w}^{k} = A\vec{w}^{k-1} = A^{k}\vec{w}^{0}$.  It may be noted that the characteristic polynomial of the matrix $A$ is identical to the characteristic polynomial for the recurrence, $p(x)$, defined above.  The assumption that $p$ have distinct roots, then implies that $A$ has distinct eigenvalues $\lambda_{\pm}$, with associated eigenvectors $\vec{v}^{\pm}$.  It follows that $\vec{w}^{0}$ can be decomposed as $\vec{w}^{0} = \alpha \vec{v}^{+} + \beta\vec{v}^{-}$, yielding $\vec{w}^{k} = A^{k}\vec{W}^{0} = \alpha\lambda_{+}^{k}\vec{v}^{+} + \beta\lambda_{-}^{k}\vec{v}^{-}$.  Then $f_{k} = w^{k}_{2} = \alpha v^{+}_{2}\lambda_{+}^{k} + \beta v^{-}_{2}\lambda_{-}^{k}$, so the solutions to the homogeneous recurrence are of the form $f_{k} = \gamma_{+}\lambda_{+}^{k} + \gamma_{-}\lambda_{-}^{k}$ for some coefficients $\gamma_{\pm}$.

For the particular solution to the inhomogeneous problem, let us postulate an ansatz of $f_{k} = b' c^{k}$.  Plugging this into the recurrence and dividing by $c^{k-2}$ yields $b'(c^{2} - a_{1}c - a_{2}) = b c^{2}$.  Since by assumption $c$ is not a root of the characteristic polynomial, $c^{2}-a_{1}c-a_{2}\neq 0$, so $f_{k} = b'c^{k}$ with $b' = b c^{2}/(c^{2} - a_{1}c - a_{2})$ is a particular solution to the inhomogeneous problem.  Then the affine space of all solutions to the inhomogeneous problem is given by 
\begin{equation}
	f_{k} = \gamma_{+}\lambda_{+}^{k} + \gamma_{-}\lambda_{-}^{k} + \frac{b c^{k+2}}{c^{2}-a_{1}c-a_{2}}
\end{equation}
for coefficients $\gamma_{\pm}$.  These two coefficients may now be determined from the given initial conditions $f_{0}$ and $f_{1}$ by solving the system of equations 
\begin{subequations}
\begin{align}
	f_{0} & = \gamma_{+} + \gamma_{-} + \frac{b c^{2}}{c^{2}-a_{1}c-a_{2}}\\
	f_{1} & = \gamma_{+}\lambda_{+} + \gamma_{-}\lambda_{-} + \frac{b c^{3}}{c^{2}-a_{1}c-a_{2}},
\end{align}
\end{subequations}
which may be rewritten
\begin{equation}
	\begin{bmatrix}1 & 1\\\lambda_{+} & \lambda_{-}\end{bmatrix}\begin{bmatrix}\gamma_{+}\\\gamma_{-}\end{bmatrix} = \begin{bmatrix}f_{0} - \frac{bc^{2}}{c^{2}-a_{1}c-a_{2}}\\f_{1} - \frac{bc^{3}}{c^{2}-a_{1}c-a_{2}}\end{bmatrix},
\end{equation}
yielding
\begin{equation}
	\gamma_{\pm} = \pm \frac{1}{\lambda_{+}-\lambda_{-}}\left(f_{1}-\lambda_{\mp}f_{0} - (c-\lambda_{\mp})\frac{bc^{2}}{c^{2}-a_{1}c-a_{2}}\right).
\end{equation}


\section{Inhomogeneity of the Linear Recurrence}
\label{app:linearRecurrenceInhom}

We seek a simple form for the inhomogeneous term in the linear recurrence for the sum $\phi(M)$ \eqref{eqn:rawmultisums}.  This requires summing both sides of \eqref{eqn:summandRecurrence} and using the fact that the right hand side leads to telescoping sums: $\sum_{k=1}^{n}\Delta_{k}X(k) = X(n+1)-X(1)$.  Since the limits of summation are different for $\phi(M-2)$, $\phi(M-1)$, and $\phi(M)$, we will sum \eqref{eqn:summandRecurrence} over $\eta,r=1,\dots, M$ and $u=1,\dots,M-2$, and add in the remaining pieces of $\phi(M-1)$ and $\phi(M)$ separately.
\begin{widetext}
\begin{subequations}
\begin{align}
(1 + & \beta + Q \beta) \xi \phi(M-2) - \big(1 + \beta + (1+Q\beta)\xi\big) \phi(M-1) + \phi(M)\nonumber\\
& = \sum_{u=1}^{M-2}\sum_{\eta=1}^{M}\sum_{r=1}^{M}\Big[\xi (1 + \beta + Q \beta)\Phi(M-2,u,\eta,r)-\big(1 + \beta + (1 + Q\beta)\xi\big) \Phi(M-1,u,\eta,r)+\Phi(M,u,\eta,r)\Big]\nonumber\\
& \qquad + \sum_{\eta,r=1}^{M}\Big[-\big(1 + \beta + (1+Q\beta)\xi\big)\Phi(M-1,M-1,\eta,r) + \Phi(M,M-1,\eta,r) + \Phi(M,M,\eta,r)\Big]\\
& = \sum_{\eta=1}^{M}\sum_{r=1}^{M}\Big[\beta \Phi(M-1, M-1, \eta, r+1) + \Phi(M-1, M-1, \eta+1, r+1) - \Phi(M, M-1, \eta+1, r+1)\Big] \nonumber\\
& \qquad + \sum_{\eta=1}^{M}\sum_{u=1}^{M-2}\Big[(1 + \beta) \xi \Phi(M-2, u, \eta, 1) - (1 + \beta + \xi) \Phi(M-1, u, \eta, 1) + \Phi(M, u, \eta, 1)\Big] \nonumber\\
& \qquad + \sum_{\eta=1}^{M}\sum_{r=1}^{M}\Big[-\big(1 + \beta + (1+Q\beta)\xi\big)\Phi(M-1,M-1,\eta,r) + \Phi(M,M-1,\eta,r) + \Phi(M,M,\eta,r)\Big]\\
& = Q\sum_{u=1}^{M-2}\sum_{\eta=1}^{M}\left[(1 + \beta) \beta^{\eta-1}\xi^{u+1}\eta \binom{M-u-2}{\eta-1} - (1 + \beta + \xi) \beta^{\eta-1}\xi^{u}\eta\binom{M-u-1}{\eta-1} + \beta^{\eta-1}\xi^{u}\eta\binom{M-u}{\eta-1}\right] \nonumber\\
& \qquad + \left[- \sum_{\eta=2}^{M-1}\beta^{\eta}\xi^{M-1}Q^{\eta}\eta\binom{M-2}{\eta-1}\right]+\left[-\big(\beta + (1+Q\beta)\xi\big)\sum_{\eta=1}^{M-1}\beta^{\eta-1}\xi^{M-1}Q^{\eta}\binom{M-2}{\eta-1}\right.\nonumber\\
& \qquad    + \left.\sum_{\eta=2}^{M}\beta^{\eta-1}\xi^{M-1}Q^{\eta-1}\eta\binom{M-2}{\eta-2} + \sum_{\eta = 1}^{M}\beta^{\eta-1}\xi^{M}Q^{\eta}\binom{M-1}{\eta-1}\right]\\
& = Q\sum_{u=1}^{M-2}\beta(1+\beta)^{M-u-2}\xi^{u}(1+\beta-\xi) + Q \beta \xi^{M-1}
 = Q\beta(1+\beta)^{M-2}\xi.
\end{align}
\end{subequations}


\section{Asymptotic Analysis of the Distance Bound}
\label{app:asymptoticAnalysis}
Expand $\beta = \beta(M)$ as
\begin{equation}\beta(M) = \frac{X}{M} + \frac{Y}{M^{2}} + O\left(\frac{1}{M^{3}}\right)\end{equation}
for some constants $X\geq0$ and $Y\geq0$.  Then
\begin{subequations}
\begin{align}
	(1+\beta)^{2} & = 1 + \frac{2X}{M} + \frac{X^{2} + 2Y}{M^{2}} + O\left(\frac{1}{M^{3}}\right)\\
	1+\beta + (1+Q\beta)\xi & = (1+\xi) + (1+Q\xi)\beta = (1+\xi) + \frac{(1+Q\xi)X}{M} + \frac{(1+Q\xi)Y}{M^{2}} + O\left(\frac{1}{M^{3}}\right)\\
	1+\beta - (1+Q\beta)\xi & = (1-\xi) + (1-Q\xi)\beta = (1-\xi) + \frac{(1-Q\xi)X}{M} + \frac{(1-Q\xi)Y}{M^{2}} + O\left(\frac{1}{M^{3}}\right)\\
	\big(1+\beta - (1+Q\beta)\xi\big)^{2} + 4 Q\xi \beta^{2} & = (1-\xi)^{2} + \frac{2(1-\xi)(1-Q\xi)X}{M} + \frac{(1+Q\xi)^{2}X^{2}+2(1-\xi)(1-Q\xi)Y}{M^{2}} +  O\left(\frac{1}{M^{3}}\right).
\end{align}
\end{subequations}
Then using the fact that
\begin{equation}\sqrt{a+x}  = \sqrt{a} + \frac{x}{2\sqrt{a}} - \frac{x^{2}}{8 a^{3/2}} + O(x^{3})\end{equation}
we can compute
\begin{subequations}
\begin{align}	
	\sqrt{\big(1+\beta - (1+Q\beta)\xi\big)^{2} + 4 Q\xi \beta^{2}} & = (1-\xi) + \frac{X(1-Q\xi)}{M} + \left[Y(1-Q\xi)+ 2QX^{2}\frac{\xi}{(1-\xi)}\right]\frac{1}{M^{2}} + O\left(\frac{1}{M^{3}}\right)\\
	\gamma_{+} & = 1 + \frac{X}{M} + \left[Y + QX^{2}\frac{\xi}{(1-\xi)}\right]\frac{1}{M^{2}} + O\left(\frac{1}{M^{3}}\right)\\
	\gamma_{-} & = \xi + \frac{QX\xi}{M} + \left[QY\xi - QX^{2}\frac{\xi}{(1-\xi)}\right]\frac{1}{M^{2}} + O\left(\frac{1}{M^{3}}\right)
\end{align}
\end{subequations}
which yields
\begin{subequations}
\begin{align}
	\gamma_{+}(\gamma_{+}-\gamma_{-}) & = \gamma_{+}\sqrt{\big(1+\beta - (1+Q\beta)\xi\big)^{2} + 4 Q\xi \beta^{2}} = (1-\xi) + \frac{X\big(2-(Q+1)\xi\big)}{M} + O\left(\frac{1}{M^{2}}\right)\\
	\gamma_{-}(\gamma_{-}-\gamma_{+}) & = -\gamma_{-}\sqrt{\big(1+\beta - (1+Q\beta)\xi\big)^{2} + 4 Q\xi \beta^{2}} = -\xi(1-\xi) - \frac{X\xi(Q+1-2Q\xi)}{M} + O\left(\frac{1}{M^{2}}\right).
\end{align}
\end{subequations}
And using the expansion
\begin{equation}	
	\frac{1}{a+x} = \frac{1}{a} - \frac{x}{a^{2}} + \frac{x^{2}}{a^{3}} + O(x^{3})\end{equation}
it may be seen that
\begin{subequations}
\begin{align}
	\frac{1}{\gamma_{+}(\gamma_{+}-\gamma_{-})} & = \frac{1}{1-\xi} - \frac{X\big(2-(Q+1)\xi\big)}{(1-\xi)^{2}M} + O\left(\frac{1}{M^{2}}\right)\\
	\frac{1}{\gamma_{-}(\gamma_{-}-\gamma_{+})} & = -\frac{1}{\xi(1-\xi)} + \frac{X(Q+1-2Q\xi)}{\xi(1-\xi)^{2}M} + O\left(\frac{1}{M^{2}}\right).
\end{align}
\end{subequations}
Moreover, 
\begin{subequations}
\begin{align}
	Q\beta\xi(\gamma_{+}+\beta) & = \frac{QX\xi}{M} + O\left(\frac{1}{M^{2}}\right)\\
	Q\beta\xi(\gamma_{-}+\beta) & = \frac{QX\xi^{2}}{M} + O\left(\frac{1}{M^{2}}\right)\\
	(1+\beta)\big[(1+\beta)-\gamma_{+}\big] & = O\left(\frac{1}{M^{2}}\right)\\
	(1+\beta)\big[(1+\beta)-\gamma_{-}\big] & = (1-\xi) + \frac{X(2-(Q+1)\xi)}{M} + O\left(\frac{1}{M^{2}}\right),
\end{align}
\end{subequations}
so, assembling the pieces, we get
\begin{subequations}
\begin{align}
	\frac{\beta A_{+}}{\gamma_{+}} & = 1 + \frac{QX\xi}{(1-\xi)M} + O\left(\frac{1}{M^{2}}\right)\\
	\frac{\beta A_{-}}{\gamma_{-}} & = -\frac{QX\xi}{(1-\xi)M} + O\left(\frac{1}{M^{2}}\right).
\end{align}
\end{subequations}
Turning now to the expressions $\gamma_{\pm}^{M}$ and $(1+\beta)^{M}$, observe that 
\begin{subequations}
\begin{align}
	e^{-\frac{X}{M}} & = 1 - \frac{X}{M} + \frac{X^{2}}{2M^{2}} + O\left(\frac{1}{M^{3}}\right)\\
	\gamma_{+}e^{-\frac{X}{M}} & = 1 + \left[Y + X^{2}\left(\frac{Q\xi}{1-\xi} - \frac{1}{2}\right)\right]\frac{1}{M^{2}} + O\left(\frac{1}{M^{3}}\right)\\
	\gamma_{+}^{M}e^{-X} & = 1 + \left[Y + X^{2}\left(\frac{Q\xi}{1-\xi} - \frac{1}{2}\right)\right]\frac{1}{M} + O\left(\frac{1}{M^{2}}\right)\\
	\gamma_{-}e^{-\frac{QX}{M}}\xi^{-1} & = 1 + \left[QY - X^{2}\left(\frac{Q}{1-\xi} + \frac{Q^{2}}{2}\right)\right]\frac{1}{M^{2}} + O\left(\frac{1}{M^{3}}\right)\\
	\gamma_{-}^{M}e^{-QX}\xi^{-M} & = 1 + \left[QY - X^{2}\left(\frac{Q}{1-\xi} + \frac{Q^{2}}{2}\right)\right]\frac{1}{M} + O\left(\frac{1}{M^{2}}\right)\\
	(1+\beta)e^{-\frac{X}{M}} & = 1 +\left[Y - \frac{X^{2}}{2}\right]\frac{1}{M^{2}} + O\left(\frac{1}{M^{3}}\right)\\
	(1+\beta)^{M}e^{-X} & = 1 +\left[Y - \frac{X^{2}}{2}\right]\frac{1}{M} + O\left(\frac{1}{M^{2}}\right)
\end{align}
\end{subequations}
Then,
\begin{subequations}
\begin{align}
	\beta A_{+}\gamma_{+}^{M-1} & = \frac{\beta A_{+}}{\gamma_{+}}\gamma_{+}^{M} =  e^{X}\left\{1 + \left[Y + Q(X+X^{2})\frac{\xi}{1-\xi} - \frac{X^{2}}{2}\right]\frac{1}{M}\right\} + O\left(\frac{1}{M^{2}}\right)\\
	\beta A_{-}\gamma_{-}^{M-1} & = \frac{\beta A_{-}}{\gamma_{-}}\gamma_{-}^{M} = O\left(\frac{\xi^{M+1}}{M}\right)\\
	(1+\beta)^{M} & = e^{X}\left\{1 + \left[Y - \frac{X^{2}}{2}\right]\frac{1}{M}\right\} + O\left(\frac{1}{M^{2}}\right)
\end{align}
\end{subequations}
and therefore
\begin{equation}\phi(M) = \frac{1}{\beta(M)}\left\{\left[Q e^{X}\big(X + X^{2}\big)\frac{\xi}{1-\xi}\right]\frac{1}{M} + O\left(\frac{1}{M^{2}}\right)\right\}.\end{equation}
Then since 
\begin{subequations}
\begin{align}
	\frac{1}{\frac{M}{X}\beta} & = \frac{1}{1 + \frac{Y/X}{M} + O\left(\frac{1}{M^{2}}\right)} = 1 - \frac{Y/X}{M} + O\left(\frac{1}{M^{2}}\right)\\
	\frac{1}{\beta} & = \frac{M}{X} - \frac{Y}{X^{2}} + O\left(\frac{1}{M}\right)
\end{align}
\end{subequations}
and
\begin{equation}\beta(M) = \begin{cases}\Gamma_{\identity}(M) = \frac{1}{Q+1}e^{\frac{\tau J_{0}}{M}}\left(e^{\frac{\tau Q J_{1}}{M}} + Qe^{-\frac{\tau J_{1}}{M}}\right) - 1  = \frac{\tau J_{0}}{M} + \frac{\tau^{2}\big(J_{0}^{2} + Q J_{1}^{2}\big)}{2M^{2}} + O\left(\frac{1}{M^{3}}\right) & J_{0}\geq J_{1}\\
\Gamma_{g}(M) = \frac{1}{Q+1}e^{\frac{\tau J_{0}}{M}}\left(e^{\frac{\tau Q J_{1}}{M}} - e^{-\frac{\tau J_{1}}{M}}\right) = \frac{\tau J_{1}}{M} + \frac{\tau^{2}\big(2J_{0}J_{1} + J_{1}^{2}(Q-1)\big)}{2M^{2}} + O\left(\frac{1}{M^{3}}\right) & J_{0}\leq J_{1},\end{cases}
\end{equation}
it follows that 
\begin{equation}
\frac{\Gamma_{g}(M)}{\beta(M)} = \begin{cases}\frac{J_{1}}{J_{0}} + O\left(\frac{1}{M}\right) & J_{0} \geq J_{1} \\ 1 & J_{0}\leq J_{1}\end{cases}.
\end{equation}
Recalling that $\xi = \zeta^{q}$ the upper bound on the ``weak'' term is given by
\begin{subequations}
\begin{align}
	\|\mc{W}\|_{1} & \leq \Gamma_{g}(M)\phi(M) \notag \\
	&= \frac{\Gamma_{g}(M)}{\beta(M)}\Big[\beta(M) A_{+}(M)\gamma_{+}^{M-1}(M) + \beta(M) A_{-}(M)\gamma_{-}^{M-1}(M) - (1+\beta(M))^{M}\Big]\\
	& = \begin{cases}\left[Qe^{\tau J_{0}}\left(\tau J_{1} + \tau^{2}J_{0}J_{1}\right)\frac{\zeta^{q}}{1-\zeta^{q}}\right]\frac{1}{M} + O\left(\frac{1}{M^{2}}\right) & J_{0} \geq J_{1}\\
&\\
\left[Qe^{\tau J_{1}}\left(\tau J_{1} + \tau^{2}J_{1}^{2}\right)\frac{\zeta^{q}}{1-\zeta^{q}}\right]\frac{1}{M} + O\left(\frac{1}{M^{2}}\right) & J_{0} \leq J_{1},\end{cases}
\end{align}
\end{subequations}
\end{widetext}
and the upper bound 
\begin{subequations}
\begin{align}
	\|B\|_{1} & \leq [1+\Gamma_{\identity}(\tau/M)\big]^{M} - e^{\tau J_{0}}\\
	& = e^{\tau J_{0}}\left[\frac{\tau^{2}Q J_{1}^{2}}{2}\right]\frac{1}{M} + O\left(\frac{1}{M^{2}}\right).
\end{align}
\end{subequations}


\section{Correlation Functions, Spectral Densities, and Bath Norms}
\label{app:correlationFunction}

Consider the Hamiltonians $H_{SB}=\sum S_{\alpha }\otimes B_{\alpha }$ and $H_{B}=I\otimes B_{0}$. The pure-bath unitary evolution operator is $U_{B}(t)=\exp (-itB_{0})$. The bath-interaction picture bath operators are 
\begin{equation}
B_{\alpha }(t)=U_{B}(t)B_{\alpha }(0)U_{B}^{\dagger }(t),\ B_{\alpha
}(0)=B_{\alpha }.
\end{equation}
Equivalently
\begin{subequations}
\begin{eqnarray}
	\dot{B}_{\alpha } &=&-i[B_{0},B_{\alpha }(t)],\ B_{\alpha }(0)=B_{\alpha } \\
	\dot{B}_{\alpha }^{\dagger } &=&-i[B_{0},B_{\alpha }^{\dagger }(t)],\ B_{\alpha }^{\dagger }(0)=B_{\alpha }^{\dagger }.
\end{eqnarray}
\end{subequations}
It follows that after differentiating $n$ times wrt $t$:
\begin{equation}
	\frac{\partial ^{n}}{\partial t^{n}}B_{\alpha }^{\dagger}(t)=(-i)^{n}[_{n}B_{0},B_{\alpha }^{\dagger }(t)],
\end{equation}
where $[_{n}B,A]:=[B,[_{n-1}B,A]]$, $[_{0}B,A]:=A$.

Now consider the correlation function
\begin{equation}
	\langle B_{\alpha }^{\dagger }(t)B_{\beta }\rangle :=\text{Tr}[\varrho_{B}B_{\alpha }^{\dagger }(t)B_{\beta }],
\end{equation}
where $\varrho_{B}$ is the initial bath state. The bath spectral density is
the Fourier transform 
\begin{equation}
	S_{\alpha \beta }(\omega )=\frac{1}{2\pi }\int_{-\infty }^{\infty}e^{i\omega t}\langle B_{\alpha}^{\dagger }(t)B_{\beta }\rangle \,\rmd t ,
\end{equation}
so that
\begin{equation}
	\langle B_{\alpha }^{\dagger }(t)B_{\beta }\rangle = \int_{-\infty }^{\infty }e^{-i\omega t}S_{\alpha \beta }(\omega ) \,\rmd\omega .
\end{equation}
Differentiating both sides $n$ times yields
\begin{equation}
	\langle \lbrack _{n}B_{0},B_{\alpha }^{\dagger }(t)]B_{\beta }\rangle=\int_{-\infty }^{\infty }\omega ^{n}e^{-i\omega t}S_{\alpha \beta }(\omega)\,\rmd\omega ,
\end{equation}
and in particular at $t=0$:
\begin{equation}
	\int_{-\infty }^{\infty }\omega ^{n}S_{\alpha \beta }(\omega )\,\rmd\omega
=\langle \lbrack _{n}B_{0},B_{\alpha }^{\dagger }]B_{\beta }\rangle .
\end{equation}

This last result allows us to bound the spectral densities in terms of the bath operator norms. One can show that $|\langle AB\rangle |\leq \left\Vert
A\right\Vert \left\Vert B\right\Vert $, where the norm is the sup-operator
norm \cite[Appendix D]{UL:10}. Thus
\bes
\begin{eqnarray}
	\left\vert \int_{-\infty }^{\infty }\omega ^{n}S_{\alpha \beta }(\omega)\,\rmd\omega \right\vert  &=&\left\vert \langle \lbrack _{n}B_{0},B_{\alpha}^{\dagger }]B_{\beta }\rangle \right\vert \\
	&\leq& \left\Vert \lbrack_{n}B_{0},B_{\alpha }^{\dagger }]\right\Vert \left\Vert B_{\beta}\right\Vert  \\
	&\leq& (2\left\Vert B_{0}\right\Vert )^{n}\left\Vert B_{\alpha }^{\dagger }\right\Vert\left\Vert B_{\beta }\right\Vert .
\end{eqnarray}
\ees
Thus, as long as $\|B_{0}\|$ and $\|B_{\alpha}\|$ are finite for all $\alpha$, the bath spectral densities $S_{\alpha\beta}(\omega)$ must decay faster than any rational function as $|\omega|\to\infty$.  Conversely, divergence of any of the moments of the spectral density implies the divergence of at least one of the bath operators. For example, all the $n\geq 1$ moments of a Lorentzian spectral density $S(\omega) = \frac{\gamma}{(\omega-\omega_0)^2+\gamma^2}$ diverge. This spectral density arises from exponentially decaying correlation functions, i.e., its Fourier transform is $F^{(0)}(t) \propto \exp(-it\omega_0-\gamma |t|)$.

\end{document}